\documentclass[a4paper,reqno]{amsart}

\usepackage{macros,cancel}
\DeclareMathOperator{\Dens}{Dens}
\DeclareMathOperator{\Ind}{Ind}
\DeclareMathOperator{\RW}{RW}

\def\ud{d_{\mc O}}
\def\lo{L^{\mc O}}

\mathchardef\mhyp="2D

\bibliography{new-citations.bib}

\usepackage{bbm}

\def\kk{{\mathbbm{k}}}

\newcommand{\hrho}{\hat{\rho}}

\allowdisplaybreaks

\title{Kodaira--Spencer theory for Courant algebroids}
\author{Julian~Kupka, Ingmar~Saberi, Charles~Strickland-Constable, Fridrich~Valach}
\address{Department of Physics, Astronomy and Mathematics, University of Hertfordshire, College Lane, Hatfield, AL10 9AB, United Kingdom}
\email{j.kupka@herts.ac.uk}
\email{c.strickland-constable@herts.ac.uk}
\email{i.saberi@herts.ac.uk}
\address{Mathematical Institute, Faculty of Mathematics and Physics, Charles University, Prague 186 75, Czech Republic}
\email{fridrich.valach@matfyz.cuni.cz}

\begin{document}

\begin{abstract}
Studying Courant algebroids on dg ringed manifolds, we observe that the associated Royten\-berg--Weinstein $L_\infty$ algebra admits a local structure reminiscent of a shifted contact structure. On a dg ringed manifold with an $n$-orientation, its symplectification produces a sheaf of $(2-n)$-shifted symplectic formal moduli problems, which we call the Courant contact model. This construction can be interpreted as a ($\Z/2\Z$-graded) theory in the Batalin--Vilkovisky formalism whenever $n$ is odd.
After developing the procedure of reduction and extension of scalars, we show how twisted backgrounds in type I supergravity naturally lead to Courant algebroids over the Dolbeault complex. Specialising to the case of a Calabi--Yau fivefold, we show that the Courant contact model for that Courant algebroid is equivalent to a central extension of minimal type I BCOV theory.
Inspired by this, we 
extend the conjecture of Costello and Li and place it within the setting of generalized geometry,
conjecturing a description of the BV formulation of type I supergravity in general twisted backgrounds.
\end{abstract}

\maketitle

\tableofcontents

\spacing{1.2}

\section{Introduction}
\label{sec:intro}

We begin with two brief, and independent, motivations to put our work in a broader context --- the first one from the perspective of shifted symplectic geometry and the second from the viewpoint of twisted supergravity. For definitions of the terms we refer the reader to the bulk of the text.

\subsection{First motivation: Courant algebroids and shifted symplectic and contact structures}
\label{ssec:mot1}
Courant algebroids can be regarded as a ``horizontal categorification'' of quadratic Lie algebras (i.e.\ Lie algebras with invariant pairing), in the sense that a Courant algebroid over a point is precisely a finite-dimensional quadratic Lie algebra \cite{liu1997manin}. From the modern perspective, an invariant pairing on a Lie algebra $\fg$ defines a $2$-shifted symplectic structure on its infinitesimal classifying space $B\fg$.

It is thus unsurprising that Courant algebroids are also closely connected to $2$-shifted symplectic geometry.
Given a Courant algebroid $E \to M$, there is a well-known procedure that constructs a (finite-dimensional) $2$-shifted symplectic dg manifold from~$E$ \cite{Roytenberg:2002nu,some-title}. This can be identified with the total space of the bundle $T^*[2]M\oplus E[1] \to M$ (though not canonically). Applying the AKSZ construction \cite{aksz} to a space $\Sigma$ with an $n$-orientation, one obtains a $(2-n)$-shifted mapping space. When $n=3$, this is a field theory in the BV formalism on $\Sigma$, called the \emph{Courant sigma model} \cite{Ikeda:2002wh, Roytenberg:2006qz}.

In this paper, we are interested in a distinct construction, albeit with many ingredients that are familiar from the standard story. The key difference is that we are interested in thinking of $E$ as a local structure on~$M$. 
Many Courant algebroids do arise in a fashion that makes this perspective natural; for example, the appropriate analogue of the Atiyah algebroid for a gerbe is a Courant algebroid, or more precisely a certain sheaf of $L_\infty$ algebras constructed from a Courant algebroid by Roytenberg and Weinstein~\cite{let,RW,BunkShahbazi}.
We will think of sections of Roytenberg and Weinstein's local Lie algebra as \emph{infinitesimal generalized diffeomorphisms}. It is natural to expect that one should be able to construct a sheaf of moduli spaces on~$M$, akin to an infinitesimal classifying space for generalized diffeomorphisms, or a horizontal categorification of the example $B\fg$ from above.

One guiding example might be the following. A Courant algebroid on a positive-dimensional manifold $M$, but with vanishing anchor map, is a bundle of quadratic Lie algebras on~$M$ with an invariant pairing valued in smooth functions. If we equip $M$ with the additional datum of a volume form, then $BE$---which is well-defined in this case, since $E$ is a local Lie algebra---has a local $2$-shifted symplectic structure.

In general, this procedure cannot work, since $BE$ does not make sense. One needs to 
enlarge one's model, placing $E$ into a local Lie algebra in which the failure of the Jacobi identity for the antisymmetrized Courant bracket is rendered exact. This construction is provided by the local Lie algebra of Roytenberg and Weinstein~\cite{RW}, denoted $\RW(E)$. Intuitively, if one thinks of the example in which~$E$ is the generalized tangent bundle $TM \oplus T^*M$, the additional generator allows us to view sections of~$T^*M$ as one-forms up to the natural gauge invariance relation in which exact one-forms are rendered trivial.

There is no natural $2$-shifted symplectic structure on the Roytenberg--Weinstein local Lie algebra, since the generator parameterizing this higher gauge invariance does not naturally pair with anything. But we have seen above that even defining the pairing on sections of~$E$ requires the introduction of a volume form. If we view this choice as an extra modulus, there is a natural pairing which makes the volume form dual to the higher gauge invariance. We can summarize this with the slogan:
\[\text{\emph{The moduli space of volume forms up to generalized diffeomorphisms is naturally 2-shifted symplectic.}}\]

There is a natural action of nonvanishing functions on the local Lie algebra describing this moduli space, by rescaling the volume form. With respect to this action, the symplectic form is exact of weight one. 
In the classical setting, symplectic structures that are homogeneous of weight one with respect to a free $\R^\times$-action correspond to contact structures on the quotient. 
One thus expects a local version of a $2$-shifted contact structure on the moduli space defined by the Roytenberg--Weinstein local Lie algebra. 
Prior works on shifted contact structures include
Grabowski in~\cite{Grabowski}, followed by Berktav~\cite{Berktav1,Berktav2,Berktav3} in the setting of derived algebraic geometry and Maglio, Tortorella, and Vitagliano~\cite{MTV,MaglioThesis} in smooth derived geometry. 
We do not develop the theory of local shifted contact structures in detail in this paper, choosing to work in the symplectic framework instead. Nevertheless, we refer to the local formal moduli problem we construct as the \emph{Courant contact model}. 

Another potential name is reflected in the title of the paper. The classical Kodaira--Spencer map relates the tangent space to the moduli space of complex structures and the degree-one cohomology of the sheaf of holomorphic vector fields.  One classical perspective on generalized geometry suggests replacing the tangent bundle by the generalized tangent bundle; our study of the relation between its sheaf of sections and natural moduli spaces motivates our choice of title.

The above considerations can be generalized to the setting of Courant algebroid structures on smooth dg ringed manifolds, defined below. In this setting, the volume form is replaced by a choice of $n$-orientation, and we obtain a local $(2-n)$-shifted  symplectic structure. 
When $n$ is odd, we obtain a $\Z/2\Z$-graded BV theory on~$M$ associated to a Courant algebroid;  in the simplest example, its fundamental fields are an $\cO$-gerbe with connective structure and an $\cO$-orientation. Just as for Chern--Simons theory, this construction gives a $\Z$-graded theory precisely when $n = 3$. However, all odd values of $n$ give interesting $\Z/2\Z$-graded theories. In this paper, we will focus on the case $n=5$ and its relations to holomorphic twists of type I supergravity.\footnote{We note that the Courant algebroid formulation applies equally well to any choice of gauge group, which we therefore do not specify further. In keeping with this, we will use the terms type I supergravity and $D=10$, $\ms N=1$ (or more properly $\ms N=(1,0)$) supergravity  interchangeably.}

\subsection{Second motivation: twisted supergravity}
\label{ssec:mot2}
In a supergravity theory, special backgrounds exist in which the ghosts for local supersymmetry transformations take nonvanishing background values. The existence of such a background is related to the existence of Killing spinors, and thus constrains the other background fields: for example, the metric must usually be a Calabi--Yau metric or otherwise exhibit reduced holonomy. Such backgrounds can be profitably studied in the Batalin--Vilkovisky formalism, but full formulations of maximal or half-maximal supergravity theories---in which the classical master equation is solved with all fermions included---were scarcely available in the literature. For type I supergravity in ten dimensions, this situation has recently been remedied in~\cite{bvsugra,fullsugra}, using techniques from generalized geometry.

In~\cite{CLBCOV1}, Costello and Li introduced a theory they called ``$(1,0)$ BCOV theory,'' defined on Calabi--Yau threefolds. 
This theory is a variant of their original ``BCOV theory,'' which builds on and generalizes the 
``Kodaira--Spencer theory of gravity''
studied by Bershadsky, Cecotti, Ooguri, and Vafa~\cite{BCOV}.
Building upon this, Costello and Li generalized their theory to define what is now called \emph{type I BCOV theory} in~\cite{CLtypeI}, making use of the unoriented B-model in topological string theory. 
Costello and Li conjectured that type I BCOV theory should be related to a target-space holomorphic twist of type I string theory; restricting to the massless sector, their conjecture implies that a subsector known as \emph{minimal type I BCOV theory} should be the holomorphic twist of type I supergravity (without additional gauge fields) in ten dimensions, possibly after minor modifications that allow for a fully nondegenerate $(-1)$-shifted symplectic structure.

Taken together, various results from the past few years come close to establishing this conjecture:
\begin{itemize}
\item Using the pure spinor superfield formalism, \cite{spinortwist} rigorously computed the twists of the eleven-dimensional and type IIB supergravity multiplets, working in the free limit on a flat background. Given the descriptions of types IIA and IIB together, this led to a computation of the free limit of holomorphic type I supergravity, which was identified with the fields of a particular \emph{potential theory} for minimal type I BCOV theory. (This is the minor modification indicated above.)
\item Candidate interactions for twisted eleven-dimensional supergravity were constructed in~\cite{RSW} and shown to pass various consistency checks. In particular, after compactification on an interval, the BV theory constructed there was shown to match a particular potential theory for minimal type I BCOV theory at the full nonlinear level.
\item Finally, \cite{CY2} proved that eleven-dimensional superspace naturally admits a Poisson structure, and that eleven-dimensional supergravity (as described in the pure spinor formalism by Cederwall~\cite{Ced-towards,Ced-11d}) is the moduli problem governed by the local Lie algebra of Hamiltonian vector fields. This result was used to construct the pure spinor description of all twists of eleven-dimensional supergravity as Poisson Chern--Simons theories, in keeping with another set of results by Costello~\cite{CostelloMtheory2}.
\end{itemize}
To complete the chain, it would suffice to prove that the Poisson Chern--Simons theory of~\cite{CY2} and the eleven-dimensional supergravity theory of~\cite{RSW} are equivalent, which would amount to a computation using the homotopy transfer theorem. 

In this paper, we will instead take a different approach, connected to the recent developments of~\cite{bvsugra,fullsugra}. Since the construction of the full supergravity theory presented there takes place in the language of generalized geometry, it is natural to expect that the theory in twisted backgrounds should also admit a description in terms of Courant algebroids. 
We confirm this expectation by connecting type I BCOV theory to the theory of Courant algebroids directly: The Courant contact model we construct here produces a natural theory in the BV formalism when the input datum is taken to be (the Dolbeault resolution of) a holomorphic Courant algebroid on a Calabi--Yau manifold of odd complex dimension. Specializing to the standard holomorphic Courant algebroid on a Calabi--Yau fivefold, we prove (Theorem~\ref{thm:equivalence}) that the Courant contact model is equivalent to the potential theory for minimal type I BCOV theory studied in~\cite{RSW}. As such, our model provides an extension of type I BCOV to the realm of generalized geometry, in particular covering flux backgrounds.

We see this as groundwork for a second, independent route to rigorously establishing the Costello--Li conjecture on  type I supergravity  in twisted backgrounds, which is in some sense more direct. The result mentioned above suggests a description of the twist at the level of the underlying Courant algebroids.
We formulate a precise statement below (Main Conjecture in \S\ref{subsec:conj}), extending Costello and Li's conjecture. 
To sketch the statement here: 
A twisted supergravity background determines an involutive isotropic subbundle $L \subset E\otimes \mb C$; 
in turn, such a datum determines a holomorphic Courant algebroid by reduction~\cite[Appendix~A]{GenKaehler}. 
We conjecture that type I supergravity in such a background is equivalent to the Courant contact model for (the Dolbeault resolution of) the holomorphic Courant algebroid obtained by reducing~$E\otimes\mb C$ along~$L$.

We have established this conjecture at the linearized level. Despite the substantial simplifications provided by the Courant algebroid formulation on both sides, though, a proof at the interacting level on general twisted backgrounds---while very much within reach---is still a nontrivial task. We therefore defer this to a second paper, currently work in progress. For the sake of exposition, we will present the linearized result in that paper as well.

\subsection{Outline}
The paper falls roughly into two parts, the first entirely independent of twisted supergravity (corresponding to the first motivation in~\S\ref{ssec:mot1}), the second, starting with \S\ref{sec:twisted-sugra-bg}, dealing with connections to BCOV theory and twisted supergravity (and thus corresponding to~\S\ref{ssec:mot2}).

We start in~\S\ref{sec:two} by reviewing some necessary background material, beginning with the notion of a dg ringed manifold, as developed by Costello and Grady--Gwilliam, followed by orientations, local Lie algebras, and shifted symplectic and contact structures. In \S\ref{sec:three} we introduce the notion of a (strict) Courant algebroid over a dg ringed manifold, discuss its basic properties and characterizations, and provide various examples. The subsequent \S\ref{sec:two-constructions} then develops two operations on Courant algebroids, namely extension of scalars and reduction, which are used in the remainder of the paper. \S\ref{sec:five} describes the Courant contact model and its relation to the Roytenberg--Weinstein $L_\infty$ algebra and contact geometry.

In the penultimate \S\ref{sec:twisted-sugra-bg} the preceding results are applied to the case of the twist of supergravity, leading to a Courant algebroid (over the Dolbeault complex) associated to a given twist datum. In particular \S\ref{subsec:cy} spells out the details of the Calabi--Yau case. Finally, in \S\ref{sec:last} we formulate the conjectures regarding the general equivalence of the resulting Courant contact model and the holomorphic twist of type I supergravity, and prove that in the Calabi--Yau case the model coincides with (a central extension of) the minimal type I BCOV theory. This establishes the match with the conjecture of Costello and Li.

\subsection{Acknowledgments} 
CSC is supported by an EPSRC New Investigator Award, grant number EP/X014959/1. FV is supported by the Charles University grant PRIMUS/25/SCI/018. This work is funded by the Deutsche Forschungsgemeinschaft (DFG, German Research Foundation) under Projektnummer 517493862 (Homologische Algebra der Supersymmetrie: Lokalit\"at, Unitarit\"at, Dualit\"at). IAS further thanks the Free State of Bavaria for support while this work was being performed. Part of the work was done while FV and IAS were in residence at Institut Mittag-Leffler in Djursholm in 2025 during the program Cohomological Aspects of Quantum Field Theory, supported by the Swedish Research Council under grant no.\ 2021-06594. The authors would like to thank Institut Mittag-Leffler, as well as the program organizers, for providing a stimulating environment for collaboration and for developing this work. Further thanks are due to the organizers of the 46th annual winter school on Geometry and Physics in Srn\'\i, Czechia; a portion of this work was completed while FV and IAS were in attendance there.
We thank A.~Ashmore, M.~Cederwall, O.~Gwilliam, F.~Hahner, S.~Raghavendran, P.~\v{S}evera, E.~E.~Svanes, and B.~Williams for valuable discussions, as well as W.~Luciani for helpful comments on the draft.
No new data was collected or generated during the course of this research.

\subsection{Notation and conventions}
Throughout, we work in a symmetric monoidal category of vector spaces, graded either by $\Z \times \Z/2\Z$ or by $\Z/2\Z$. In the first setting, the integer grading is called ``cohomological degree'' or ``ghost number,'' and the $\Z/2\Z$ factor is called ``fermion parity'' or ``intrinsic parity.''  $\Pi$ denotes the parity shift functor, and square brackets denote shifts of the integer grading; we adhere to the convention $V[i]^j \coloneqq V^{i+j}$. The Koszul rule of signs is determined by the sum of the degree and the intrinsic parity modulo two; this totalization is just called ``parity,'' and is denoted with the symbol $|\slot|$. In the $\Z/2\Z$-graded setting, only the (total) parity is defined. 
When writing a formula for an object valued in a free module over a cdga, we frequently leave the internal Koszul sign implicit, writing explicit formulas with the signs appropriate for generators of the module.
 Our conventions are cohomological throughout; every differential has degree $(+1,+)$. Duals reverse grading. The Hamiltonian vector field $X_H$ of a function $H$ is defined by $i_{X_H}\omega=-dH$. Further conventions can be found in~\cite{Kupka:2025ulg}.

\section{Background}
\label{sec:two}

\subsection{Dg ringed manifolds}
\label{ssec:dgringed}
In~\cite{KevinWitten}, Costello suggests an approach to derived differential geometry in which a particularly nice class of spaces, called \emph{(smooth) dg ringed manifolds}, plays an important role. 
(The theory was further developed in~\cite{GG1,GG2}, under the closely-related nomenclature \emph{$L_\infty$ spaces}.) 
In this paper, we are interested in defining Courant algebroids over such spaces. We quickly review some definitions to set the stage and to fix terminology.

Fix a base field $\field = \R$ or $\C$.  Recall that a \emph{cdga} is a unital $\Z$-graded commutative algebra over $\field$, equipped with a derivation of degree $+1$, denoted $\bar\partial$ and called the \emph{differential}. A \emph{morphism} of two cdgas is a morphism of algebras preserving both the grading and the differential. If $A$ is a cdga, a \emph{cdga over $A$} is a cdga $B$ equipped with a morphism $A \to B$.
Given a cdga $(A,\dbar)$, we can forget the differential to obtain another cdga, denoted $A^\#$. 

We now recall Costello's definition.
\begin{dfn}[{\cite[Definition~1.0.1]{KevinWitten}}]
\label{def:dgringed}
A \emph{smooth dg ringed manifold} (over $\field$) consists of the data of:
\begin{itemize}
\item a smooth $d$-dimensional manifold $M$, with sheaf of smooth $\field$-valued functions denoted $C^\infty$ and sheaf of smooth de Rham forms denoted $\Omega^\bu$;
\item a sheaf $\cO$ of cdga's on~$M$, with differential denoted $\dbar$, and further equipped with maps
\[
\begin{tikzcd}
\Omega^\bu \ar[r, "j"] & \cO \ar[r, "q"] & C^\infty.
\end{tikzcd}
\]
of cdga's.
\end{itemize}
These data are required to satisfy the following properties:
\begin{itemize}
\item $\cO^\#$ is a locally free module of finite total rank over~$\Omega^0 = C^\infty$;
\item The kernel of~$q$ is a sheaf of nilpotent ideals in~$\cO$;
\item The cohomology of~$\cO(U)$ is concentrated in nonpositive degrees over any sufficiently small open subset of~$U$.
\end{itemize}
\end{dfn}
A map of smooth dg ringed manifolds is defined just as for a map of ringed spaces in general. We now recall some important examples from~\cite{KevinWitten}.
\begin{eg}
\label{eg:Msm}
Any smooth manifold $M$ is a smooth dg ringed manifold by taking $\cO = C^\infty$, with zero differential and concentrated in degree zero. The map $j$ is defined by quotienting by the sheaf of nilpotent ideals $\Omega^{\geq 1} \subset \Omega^\bu$. This smooth dg ringed manifold is again denoted $M$.
\end{eg}
\begin{eg}
\label{eg:MdR}
There is another natural structure on $M$, defined by taking $\cO = \Omega^\bu$ with $j$ the identity map and $q$ the quotient map identified above. We call the resulting smooth dg ringed manifold $M_\dR$.
\end{eg}
\begin{eg}
\label{eg:Mdbar}
If $M$ is a complex manifold, we can equip it with 
the sheaf of Dolbeault forms $\Omega^{0,\bullet}$, with $\bar\partial$ being the Dolbeault operator. We denote the resulting smooth dg ringed manifold by $M_\dbar$. 
Note that the holomorphic Poincar\'e lemma shows that $\cO$ is locally quasi-isomorphic to the sheaf $\cC_\hol$ of holomorphic functions, placed in degree zero. 
  \end{eg}
  
We will consider sheaves of dg~$\cO$-modules over dg ringed manifolds. Recall that a dg module over a cdga $A$ is a $\Z$-graded cochain complex $N$, together with a degree-preserving cochain map $A \otimes_\field N \to N$ making $N^\#$ a left $A^\#$-module. 
 An \emph{inner product} on a dg $A$-module is a self-dual $A$-linear isomorphism from the module to its dual.
In this paper, we only deal with sheaves of modules that are locally free of finite rank over~$\Omega^0$ after forgetting the differential; such dg modules are termed \emph{semifree}.
As an example, we note that $\Der\cO$ is canonically such a sheaf of~$\cO$-modules on any dg ringed manifold.

\begin{rmk}
\label{rmk:homdeg}
While all conditions in the definition quoted above apply to the examples we treat in this paper, we make no essential use of several of the restrictions Costello imposes. In particular, the third condition in Definition~\ref{def:dgringed}, requiring cohomology in positive degrees to vanish locally, plays no role for us, and the finite-rank condition appears only at a technical level. We fully anticipate that our results will carry over, with appropriate minor modifications, to a setting where the third condition is eliminated and the first condition is weakened to require only degree-wise finiteness for the $\Omega^0$-module $\cO^\#$.
\end{rmk}

On a formal level, the same definition can be taken to apply, \emph{mutatis  mutandis}, to any locally ringed space $(M,\cC)$, where $\cC$ replaces the smooth functions and $\Omega^\bu$ is interpreted as the algebraic de Rham complex $\Omega^\bu_{\cC/\field}$ (as defined, for example, in~\cite[\S17.30]{Stacks}). In full generality, this notion will not in general be as useful---in particular, since $\cC$ need not be a fine sheaf. To emphasize, we do not generalize Costello's construction of a derived stack from a dg ringed manifold to such a setting.

We will nevertheless work with instances of such a generalization. Most importantly, we define a \emph{complex dg ringed manifold} (necessarily over~$\C$) to be given by Definition~\ref{def:dgringed}, but where $M$ is a complex manifold and $\cC = \cC_\hol$ is the sheaf of holomorphic functions. Then we can repeat~Example~\ref{eg:Msm} to view $(M,\cC_\hol)$ as a complex dg ringed manifold, which we denote $M_\hol$. 

Our use of this example is essentially for technical reasons. Indeed, Dolbeault resolution defines a functor from complex dg ringed manifolds to smooth dg ringed manifolds over~$\C$; this functor  carries $M_\hol$ to $M_\dbar$. We will need both notions in order to discuss the Dolbeault resolution of a holomorphic Courant algebroid, which is one of our aims in the sequel.

Since a supermanifold is a locally ringed space, the observation above equally well allows us to define the notion of a \emph{smooth dg ringed supermanifold}. Taking $M$ to be a smooth supermanifold, Examples~\ref{eg:Msm} and~\ref{eg:MdR} apply out of the box. If $M$ is equipped with a superconformal structure, the corresponding \emph{pure spinor superspace}, as defined in~\cite[{\S2}]{SCA}, is naturally a smooth dg ringed supermanifold that plays the role of~$M_\dbar$ (if one omits the third restriction from~Definition~\ref{def:dgringed} as discussed in Remark~\ref{rmk:homdeg}). 

For other, more thorough treatments of related background material, we refer the reader to the discussion of dg manifolds in~\cite[{\S2}]{Carchedi} (though note that the discussion there is focused on the nonpositively graded case), as well as to~\cite{Nuiten-thesis}. 

\subsection{Orientations}
We will require a notion of orientation which, at least roughly speaking, is a simple special case of the notion of $\cO$-orientation introduced in~\cite[Definition~2.4]{PTVV}. (Compare also to Ginzburg's notion of a Calabi--Yau $n$-algebra~\cite{Ginzburg}.)

Recall that we required, for any smooth dg ringed manifold, that $\cO^\#$ be a locally free sheaf of~$C^\infty$-modules of finite rank. Such an object is the smooth sections of a finite-dimensional (graded) vector bundle. To any vector bundle $E$ on a smooth manifold~$M$ is associated a bundle
\[
E^! \coloneqq E^\vee \otimes \Dens(M),
\]
where $E^\vee$ is the dual in the category of vector bundles over~$M$ and $\Dens(M)$ is the bundle of densities, which is the dualizing complex (in the sense of Verdier duality) in this setting. 
$E^!$ is called the \emph{shriek dual} of~$E$, and is characterized by the property that
sections of~$E$ are in duality with compactly supported distributional sections of~$E^!$ over the ground field~$\field$. This construction extends to cochain complexes of vector bundles, and thus to~$\cO$.

\begin{dfn}
Let $(M,\cO)$ be 
a smooth dg ringed manifold. An \emph{$n$-orientation} is the datum of a quasi-isomorphism
\[
\omega\colon \cO[n] \to \cO^!
\]
of sheaves of dg $\cO$-modules. An {$n$-orientation} is called \emph{strict} when it is an isomorphism of cochain complexes.
\label{dfn:or}
\end{dfn}

\begin{eg} 
On a smooth manifold $(M,C^\infty)$, we see that $\cO^!$ is the sheaf of densities supported in degree zero. Thus an orientation (here, a $0$-orientation) corresponds to a choice of a globally nonvanishing density (which always exists).
\end{eg}

\begin{eg}
For the de Rham space $M_\dR$ of an oriented smooth $d$-manifold, we note that
\[
\left( \Omega^\bu \right)^! \cong \Omega^\bu [d].
\]
Thus $M_\dR$ admits a strict $d$-orientation, which consists of a choice of globally nonvanishing function. In fact, the unit gives a canonical choice.
\end{eg}

\begin{eg}
Let $M$ be a complex manifold of real dimension $d = 2n$. Then 
\[
\left( \Omega^{0,\bu} \right)^! = \Omega^{n,\bu}[n].
\]
Thus $M_\dbar$ admits an $n$-orientation precisely when $M$ is Calabi--Yau, 
and the datum of the orientation is a choice of globally nonvanishing holomorphic top form.
\end{eg}
\begin{rmk}
In this paper, a \emph{Calabi--Yau manifold} will mean a complex $n$-manifold equipped with an $n$-orientation (holomorphic volume form). The notion of a Ricci-flat K\"ahler metric will play no direct role, and a ``Calabi--Yau structure'' does not involve the choice of such a metric.
\end{rmk}
Pure spinor superspace admits an orientation precisely when the variety of square-zero elements in the supertranslation algebra is Gorenstein~\cite{perspectives}; in general, this orientation is not strict.

\subsection{Local Lie algebras}
\label{ssec:locLie}

The notion of \emph{local Lie algebra} (or, synonymously for us, local $L_\infty$ algebra) given in~\cite{CG2} allows one to work with sheaves of formal moduli problems over a smooth manifold~$M$ (``local moduli problems''). This notion is appropriate for the study of classical perturbative field theories, infinitesimal symmetries, perturbations of background fields, and the like: it incorporates the physical requirement of locality with a derived approach to the geometry of (formal) moduli spaces. We recall this definition here, with a slight modification that places it in the setting of~\S\ref{ssec:dgringed}.
\begin{dfn}[{Compare~\cite[Definition~3.1.3.1]{CG2}}]\label{dfn:local-lie-algebra}
Let $(M,\cO)$ be a smooth dg ringed manifold. A \emph{local Lie algebra over~$\cO$} consists of:
\begin{itemize}
\item a sheaf $\cL$ of dg $\cO$-modules, with $\cL^\#$ locally free of finite rank over~$C^\infty$;
\item a collection of graded-antisymmetric polydifferential operators
\[
\mu_n\colon {\cL}^{\otimes n} \to \cL[2-n]
\]
for $n \geq 2$, which, together with the differential, give $\cL$ 
the structure of a sheaf of $L_\infty$ algebras.
\end{itemize}
\end{dfn}
Given the assumptions on~$\cO$, any local Lie algebra over~$\cO$ is automatically a local Lie algebra over~$C^\infty$, and thus in the sense of~\cite{CG2}. The definition applies, \emph{mutatis mutandis}, to any the setting with $\cC$ fine; thus, for example, to smooth supermanifolds. 
\begin{rmk}
\label{rmk:resolved}
We emphasize that, in the context of applications to field theory, the definition of a local Lie algebra requires a presentation by unconstrained free fields. All gauge symmetries and all constraints must be appropriately resolved by the introduction of ghosts, antifields, or analogous generators in nonzero cohomological degree. This requirement ensures that field theories globalize to arbitrary spacetimes correctly, without needing to take care of issues related to sheaf cohomology.
\end{rmk}

\begin{rmk}
The fundamental theorem of derived deformation theory gives an equivalence between formal moduli problems and $L_\infty$ algebras.
(See~\cite{Pridham,DAGX}, building on earlier ideas of  Deligne and Drinfeld, among others.) This equivalence underlies the approach to perturbative field theories and their symmetries developed in~\cite{CG1,CG2}. We do not give a full review here, but offer a very brief recollection to fix the terminology and notation we will use in the sequel; more details are in~\cite[Appendix~A]{CG2}, and in the further literature.

 Following the exposition in~\cite[Appendix~A]{ButsonYoo}, we can capture the fundamental ingredients of this theorem in a diagram for any derived stack $X$ and any choice of point $x \in X$:
\begin{equation}
\begin{tikzcd}
& & X \ar[dr] \ar[dl] & &   \\
\text{(formal derived stack)} & \widehat{X}_x \ar[rr, shift left, "{\mathbb{T}[-1]}"] & & \mathbb{T}_x[-1] X \ar[ll, shift left, "B"] & \text{($L_\infty$ algebra)}
\end{tikzcd}
\label{eq:dStk}
\end{equation}
The $-1$-shifted tangent complex at $x$, $\mathbb{T}_x[-1]X$, carries a natural $L_\infty$ algebra structure. By definition, a formal derived stack has a  unique point, so the operation that sends a formal derived stack to the $-1$-shifted tangent complex at the basepoint is canonical. This is the arrow from left to right in the bottom row of~\eqref{eq:dStk}.

 To go from right to left, we use the \emph{functor of Maurer--Cartan elements} (here denoted $B$) as studied in~\cite{Hinich,Getzler}. 
 The technical details of this construction will not be important for us in what follows, and the reader can safely ignore them, but we mention them to highlight a subtlety with grading conventions.
 A formal derived stack is presented by its functor of points, which is a functor from dg Artin algebras to simplicial sets satisfying certain descent properties. 
 To an Artin algebra $A$ with maximal ideal $\fm$, the functor of Maurer--Cartan elements $B\fg$ assigns the simplicial set whose $n$-simplices are locally constant solutions of the Maurer--Cartan equation over the standard $n$-simplex: 
 \[
x \in  \fg \otimes \fm \otimes \Omega^\bu(\Delta^n) \colon  \sum_{n \geq 1} \frac{1}{n!} \mu_n(x,\ldots,x) = 0.
 \]
The linearization $\mathbb{T}(B\fg)$, which is the space of fields of a field theory in the BV formalism, is canonically identified with~$\fg[1]$.
Correspondingly, the observables of such a field theory, which are \emph{functions} on the space of fields $\fg[1]$ equipped with the BV differential, are identified with the Chevalley--Eilenberg cochains $C^\bu(\fg)$. One is free to think of $C^\bu(\fg)$ as presenting a graded manifold equipped with a homological vector field, as one prefers.
While both conventions may appear in the sequel, we will try and consistently use the notation and terminology established here to disambiguate.
\end{rmk}

We also need the notion of a vector bundle on a local moduli problem. 
Given a local Lie algebra $\cL$ on~$M$, such a vector bundle consists of the data of a locally semifree sheaf $\ms V$ of dg $\cO$-modules on~$M$, together with a collection of polydifferential operators 
\[
\rho^{(n)}\colon \cL^{\otimes n} \otimes \ms V \to \ms V 
\] 
for $n \geq 1$, which (together with the differential) make $\ms V$ into a sheaf of $L_\infty$-modules for $\cL$. Over any open set $U$, the Chevalley--Eilenberg differential with coefficients in this module makes $C^\bu(\cL,\ms V)$ into a semifree dg module over $C^\bu(\cL)$, which we think of as the sections of the bundle over the moduli problem on~$U$.
Examples include the tangent and cotangent bundles of a local moduli problem, which are presented by 
taking $\ms V = \cL[1]$ or $\ms V = \cL^![-1]$, respectively. We refer the reader to~\cite[\S2]{ButsonYoo} for further details.

Any symmetry of a field theory can be presented by a local Lie algebra $\cG$. The corresponding coupling lets us view the theory as a family over the sheaf of formal spaces $B\cG$, which one interprets as a ``coupling to background fields.''  Moreover, any perturbative field theory in the Batalin--Vilkovisky formalism includes the datum of a local Lie algebra $\cL$; the moduli space of on-shell field configurations is then $B\cL$, and the space of BV fields is $\cL[1]$. In full, perturbative (Lagrangian) field theories are described by certain special local Lie algebras equipped with one additional structure, to which we turn now.

\subsection{Shifted symplectic and contact structures}
The formal study of shifted symplectic structures in the context of derived algebraic geometry originated in~\cite{PTVV}, but the first example of such a structure was the \emph{antibracket} appearing in the Batalin--Vilkovisky formalism~\cite{BV1,BV2}, and there is a large body of prior work in the physics literature that we cannot hope to fully cite here. 

Although the notion of shifted symplectic structure is model-independent, in the physics context one typically works with (a local version of) a \emph{strict} shifted symplectic structure on the space of fields $\cL[1]$. 
Such a structure is defined by a graded-alternating isomorphism
\[
\omega\colon \cL[1] \to (\cL[1])^![-1] \cong \cL^![-2]
\]
of sheaves of dg $\cO$-modules on~$M$. 
This is all we will need here, and so we do not review this notion in full. It is important to note, though, that many weakenings of this notion are possible. $\omega$ may be a quasi-isomorphism that is not invertible at the cochain level; it may be a differential operator; or it may be replaced by a more degenerate object such as a shifted presymplectic or Poisson structure. 

A notion of shifted contact structure was defined by Grabowski in~\cite{Grabowski}, which also studies a notion of \emph{contact Courant algebroid}. Grabowski proves that his notion, which arises from a generalization of the equivalence between Courant algebroid structures on $E \to M$ and Hamiltonian homological vector fields on a $2$-shifted symplectic dg manifold constructed from~$E$ to the contact setting, reduces to a generalization of the usual Courant algebroid axioms in which the pairing is valued in a line bundle on~$M$. This structure was previously studied by~\cite{Baraglia}, who called it a ``conformal Courant algebroid.''
We expect that our construction generalizes straightforwardly to this setting, but do not pursue this explicitly here.

More recently, shifted contact structures were studied in the context of derived algebraic geometry in~\cite{Berktav1,Berktav2,Berktav3},
and in smooth derived geometry in~\cite{MTV,MaglioThesis}.
Such a structure on a derived stack $X$ can be presented in many ways, but can be thought of as a map from the tangent complex of~$X$ to an $m$-shifted line bundle on~$X$, satisfying an appropriate nondegeneracy condition. 
In the sequel, we will think of our construction as closely related to a version of these constructions, but again set up in a fashion that is local over some smooth spacetime $M$. We emphasize, though, that our results are independent of this interpretation, and that we do not develop any foundational theory for such local contact structures in this paper.

\section{Courant algebroids}
\label{sec:three}

\subsection{Courant algebroids on dg ringed manifolds}
  We start by introducing the central structure of our interest.
  First, recall that there is a natural map $\ud\colon \cO \to (\Der\cO)^\vee$, where $\Der \cO$ denotes the left $\cO$-module of derivations $\cO$.\footnote{The assumptions we have made on~$\cO$ above make $\Der \cO$ into a model for the tangent complex in the examples we consider, so there are no derived subtleties to worry about here.} This can be defined by requiring the relation
  \[
      i_X\ud f = (-1)^{|X|} Xf 
  \]
  for all $f \in \cO$ and $X \in \Der\cO$, 
 in keeping with the conventions of~\cite{Kupka:2025ulg}. The map $\ud$ extends as a derivation to the exterior algebra of $(\Der\mc O)^\vee$ over $\mc O$.
  \begin{dfn}
  \label{dfn:CA}
    Let $(M,\mc O)$ be a dg ringed manifold. A (strict) \emph{Courant algebroid over~$\cO$} is given by the following data:
    \begin{itemize}
      \item a sheaf $\mc E$ of dg $\mc O$-modules (with $\cE^\#$ locally free of finite rank over~$\Omega^0$);
      \item an inner product $\la\slot,\slot\ra$, denoted $\eta$ in its guise as a self-dual isomorphism $\cE \to \cE^\vee$;
      \item an $\cO$-linear map $\rho\colon \mc E\to \on{Der}\mc O$, called the \emph{anchor map};
      \item a map of vector spaces $[\slot,\slot]\colon \mc E\otimes_{\field}\mc E\to\mc E$.
    \end{itemize}
    These data are required to satisfy the following properties for all homogeneous $u,v,w\in\mc E$, $f\in\mc O$:
    \begin{enumerate}
            \item ``Leibniz:'' \label{it:Leibniz}
              \[ [u,fv] = (-1)^{|f||u|}f [u,v] + (-1)^{u}(\rho(u)f) v. \]
                \item ``Failed antisymmetry:'' \label{it:failed-anti-symmetry}
                  \[ [u,v] + (-1)^{|v||u|}[v,u] = \cD \la u,v \ra, \]
                  where $\cD\colon \Omega^0 \to \cE$ is the map given by the composition $\eta^{-1} \rho^\vee \ud$, implying $\inner{u, \cD f} = \rho(u)f$.
                \item ``Pairing compatibility:'' \label{it:pairing-compatibility}
                  \[ \rho(u)  \la v,w \ra = \la [u,v], w \ra + (-1)^{|u||v|}\la v, [u,w] \ra. \]
                \item ``Jacobi:'' \label{it:Jacobi}
                  \[ [u,[v,w]]=[[u,v],w]+(-1)^{|u||v|}[v,[u,w]].\]
  \end{enumerate}
That is, the adjoint action of an element of $\cE$ by bracketing on the left is itself a derivation of the bracket.
  \end{dfn}

  We note that our definition is a small variant of~\cite[Definition 2.8]{RoytenbergCDA}, where a notion of Courant algebroid on a ringed space is defined. In Roytenberg's definition, one studies a sheaf of Courant--Dorfman algebras over~$\cO$, requiring neither local freeness of $\cE$ nor nondegeneracy of~$\eta$. Here, we allow $\mc O$ to be a dg ring and $\cE$ a dg module, but impose nondegeneracy  of~$\eta$ and local semifreeness. In the broader literature, ``Courant algebroid'' often refers to a vector bundle rather than to its sheaf of sections, but the Serre--Swan theorem makes this distinction immaterial~\cite{Swan}, and Roytenberg's definition is more general in any case.

\begin{rmk}
\label{rmk:strict}
In this paper, we work only with a strict notion of Courant algebroid, and with explicit models. We do not give any obviously model-independent definition of a Courant algebroid structure here. Nevertheless, the primary objects of study are various local $L_\infty$ algebras associated to such Courant algebroids, for which the necessary foundations are well-established. We thus allow ourselves to implicitly use equivalences of local $L_\infty$ algebras; for example, when $M$ is a complex manifold, we will not distinguish between $\Der \Omega^{0,\bu}(M)$ and $\Omega^{0,\bu}(M,T^{1,0}M)$. It is also reasonable to expect that the logic could be reversed, taking the perspective that formal moduli problems equipped with appropriate shifted symplectic or shifted contact structures can be used to \emph{define} a model-independent notion of Courant algebroid. 
\end{rmk}

There is an obvious parallel notion of a Lie algebroid over~$\cO$, to which we will make reference freely in the sequel.

\begin{eg}\label{dfn:standard}
  For any dg ringed manifold $(M,\mc O)$, there is a \emph{standard Courant algebroid}, defined by taking 
  \[\mc E\coloneqq \Der\mc O\oplus(\Der\mc O)^\vee.\]
  We equip this with the structure of a Courant algebroid, defined by the following formulas on homogeneous elements, which we denote by $x + \alpha$ (with $|x| = |\alpha|$) or similar. The pairing and the anchor are given by
  \[ \la x+\alpha,y+\beta\ra=(-1)^{|x||y|}\alpha(y)+\beta(x),\qquad \rho(x+\alpha)=(-1)^{|x|}x,\]
and the bracket is given by
  \[[x+\alpha,y+\beta]=L^{\mc O}_xy+L^{\mc O}_x\beta-(-1)^{(|x|+1)|y|}i_y\ud\alpha.\]
  Here $L^{\mc O}_xy$ is the commutator of derivations, and $L^{\mc O}_x\beta$ is defined by requiring the pairing between $\Der\mc O$ and $(\Der\mc O)^\vee$ to be invariant.
\end{eg}

 \begin{dfn}
  Given a Courant algebroid $\mc E$ over $\mc O$ and $u\in \mc E$, we define the \emph{generalized Lie derivative} $\ms L_u$ by its action on sections of $\mc O$ or $\mc E$, given by 
  \[\ms L_uf=\rho(u)f,\qquad \ms L_uv=[u,v],\qquad \forall f\in \mc O,\quad v\in\mc E.\]
  This naturally extends as a derivation to tensor products and (shriek) duals, in the latter case by requiring
  \[\int_Mf\ms L_u\mu=-(-1)^{|f|\cdot|u|}\int_M(\ms L_uf)\mu,\]
  for all $\mu\in \mc O^!$ and compactly supported $f\in \mc O$.
 \end{dfn}

\begin{rmk}
  In the remainder of the paper, we will only consider $\cO$-Courant algebroids whose generators are concentrated in degree zero. As such, and for the sake of readability, we will henceforth always write formulas as if all objects appearing in them were even. The generalization to arbitrary parity is straightforward.
\end{rmk}

\begin{lem}
\label{lem:Leibmap}
Let $\cE$ be a Courant algebroid over $\mc O$. Then $\rho$ preserves the brackets, i.e.
\[\rho([u,v])=[\rho(u),\rho(v)],\qquad \forall u,v\in\mc E.\]
Furthermore, the following is a (self-dual) complex of $\cO$-modules:
    \begin{equation}\label{ses}
        \begin{tikzcd}
            0 \arrow[r] & \Der\cO^\vee \arrow[r, "\eta^{-1} \rho^\vee"] & \cE \arrow[r, "\rho"] & \Der\cO \arrow[r] & 0.
        \end{tikzcd}
    \end{equation}
\end{lem}
\begin{proof}
The first claim is standard and follows by repeatedly using the Leibniz property in the Jacobi identity for $u$, $v$, and $fw$. Applying $\rho$ to $[u,v]+[v,u]=\mc D\la u,v\ra$ we then obtain $\rho\circ\mc D=0$.
Since it factors through the derivation $\ud$ as
\[
\rho \circ \cD = \rho\, \eta^{-1} \rho^\vee \ud,
\]
we conclude that the linear map $\rho \,\eta^{-1} \rho^\vee$ is zero.
\end{proof}
Roytenberg calls this sequence the {tangent complex} of~$\cE$~\cite[(2.4)]{RoytenbergCDA}.
  \begin{dfn}
    If the sequence \eqref{ses} is exact, we say that the Courant algebroid is \emph{exact}. More generally, a Courant algebroid is \emph{transitive} if the anchor map $\rho$ is surjective.    
 \end{dfn}

\begin{rmk}
 Recall that a vector subspace $C \subset (E, \eta)$ is called \emph{coisotropic} if $C^\perp\subset C$. Taking $q\colon E \to E/C$ to be the quotient map, this condition can be rewritten as $\im(\eta^{-1} q^\vee) \subset \ker(q)$, i.e.\ $q\eta^{-1}q^\vee=0$. Applying the first isomorphism theorem, the second statement of Lemma \ref{lem:Leibmap} is thus equivalent to saying that 
 \[\text{$\ker\rho$ is coisotropic.}\]
\end{rmk}

  \begin{eg}\label{ex:transitive-smooth}
    If $\mc O$ is the sheaf $C^\infty$ of $\mb R$-valued smooth functions on a manifold then Definition \ref{dfn:CA} recovers the standard definition of a \emph{(smooth) Courant algebroid} \cite{liu1997manin}. In other words, $\mc E$ correspods to sections of a real vector bundle $E\to M$, further equipped with a fibrewise inner product, a vector bundle map $E\to TM$, and a bracket on $\Gamma(E)$ satisfying the axioms \eqref{it:Leibniz}--\eqref{it:Jacobi}.
    
    A simple example of a smooth Courant algebroid structure can be constructed starting with a flat principal $G$-bundle over $M$, with the real Lie algebra $\mf g=\on{Lie}(G)$ carrying a nondegenerate invariant inner product $\la\slot,\slot\ra_\mf g$. Denoting the adjoint vector bundle by $\on{ad}_G$, the Courant algebroid is given by
    \[\mc E=\Gamma(TM\oplus T^*M\oplus \on{ad}_G),\qquad \la x+\alpha+s,y+\beta+t\ra=\alpha(y)+\beta(x)+\la s,t\ra_\mf g,\qquad \rho(x+\alpha+s)=x,\]
    with the bracket
    \[[x+\alpha+s,y+\beta+t]=L_xy+(L_x\beta-i_yd\alpha+\la \nabla s,t\ra_\mf g)+(\nabla_xt-\nabla_ys+[s,t]_\mf g),\]
    where $L$ stands for the Lie derivative. In fact, every transitive Courant algebroid locally takes this form~\cite{let}. Indeed, since the result is local and principal bundles are locally trivial, it is sufficient to consider the trivial principal $G$-bundle as a local model. If~$G$ is trivial, this reduces to the \emph{standard smooth Courant algebroid} on $M$ in the sense of Example~\ref{dfn:standard}.
  \end{eg}
  \begin{eg}
    As a small modification, we can take $\field=\mb C$ and $\mc O=C^\infty_\mb C$ to be the sheaf of smooth $\mb C$-valued functions on $M$. A Courant algebroid over $\mc O$ is then described by a complex vector bundle equipped with a $\mb C$-valued pairing, a $\mb C$-bilinear map on its sheaf of sections, and an anchor map $E\to TM\otimes \mb C$, satisfying the usual axioms. 
  \end{eg}
  \begin{eg}
    If $\mc O$ is the sheaf of holomorphic functions on a complex manifold, Definition \ref{dfn:CA} recovers the definition of a \emph{holomorphic Courant algebroid} \cite{GenKaehler}, i.e.\ a holomorphic vector bundle $E\to M$ over a complex manifold, equipped with a fibrewise inner product, a holomorphic map $E\to T^{1,0}M$, and a bracket on the sheaf $\mc E$ of holomorphic sections of $E$ satisfying the axioms \eqref{it:Leibniz}--\eqref{it:Jacobi}.
    
    An example of a holomorphic Courant algebroid\footnote{This construction was introduced in \cite{Garcia-Fernandez:2018ypt} under the name \emph{string algebroids}.} can be constructed starting from a flat holomorphic principal $G$-bundle, with (the complex Lie algebra) $\mf g$ carrying an invariant inner product $\la\slot,\slot\ra_\mf g$. The sheaf $\mc E$ then consists of holomorphic sections of $T^{1,0}M\oplus T^{*1,0}M\oplus \on{ad}_G$, with $\la\slot,\slot\ra$ and $\rho$ of the same form as in Example~\ref{ex:transitive-smooth}, and with the bracket
    \[[x+\alpha+s,y+\beta+t]=L_xy+(L_x\beta -i_y\partial\alpha+\la \nabla s,t\ra_\mf g)+(\nabla_xt-\nabla_ys+[s,t]_\mf g).\]
    Again, in the case of the trivial group $G$ this is the \emph{standard holomorphic Courant algebroid}.
  \end{eg}
  
  \begin{eg}\label{ex:long}
    Let again $M$ be a complex manifold. The standard Courant algebroid over $M_{\bar\partial}$ (i.e.\ with $\mc O=(\Omega^{0,\bullet},\bar\partial)$) then has the form
    \[\mc E=(\Omega^{0,\bullet}(T^{1,0}M\oplus T^{*1,0}M),\bar\partial),\]
    where for any $\zeta,\xi,\vartheta,\kappa\in \Omega^{0,\bullet}$, $y,z\in \Gamma(T^{1,0}M)$, and $\alpha,\beta\in \Gamma(T^{*1,0}M)$ we have
    \begin{align*}
      \rho(\zeta\otimes y+\xi\otimes \alpha)\vartheta&=\zeta \wedge L_y\vartheta\\
      \la \zeta\otimes y+\xi\otimes \alpha,\vartheta\otimes z+\kappa\otimes \beta\ra&=\zeta\wedge \kappa \,\beta(y)+\xi\wedge \vartheta \,\alpha(z)\\
      [\zeta\otimes y,\vartheta\otimes z]&=\zeta\wedge\vartheta \otimes L_yz+\zeta\wedge(L_y\vartheta)\otimes z-(L_z\zeta)\wedge \vartheta\otimes y\\            
      [\zeta\otimes y,\kappa\otimes \beta]&=\zeta\wedge\kappa \otimes (\partial i_y +i_y\partial)\beta +\zeta\wedge(L_y\kappa)\otimes \beta+(\mc D\zeta)\wedge \kappa\,\beta(y)\\    
      [\xi\otimes \alpha,\vartheta\otimes z]&=\xi\wedge\vartheta \otimes i_z\partial \alpha-(L_z\xi)\wedge \vartheta\otimes\alpha+(\mc D\xi)\wedge\vartheta\,\alpha(z)\\    
      [\xi\otimes \alpha,\kappa\otimes \beta]&=0,
    \end{align*}
    where $\mc D\vartheta$ is given by $(-1)^{\deg\vartheta}\partial \vartheta$ followed by the natural identification $\Omega^{1,\bullet}\cong\Omega^{0,\bullet}\wedge\Omega^{1,0}$. Note that we have replaced $\Der \Omega^{0,\bu}$ by the quasi-isomorphic local dg Lie algebra $\Omega^{0,\bu}(T^{1,0})$, as indicated above in Remark~\ref{rmk:strict}.
  \end{eg}

 \subsection{Characterizing Courant algebroids}\label{ssec:weak}
In practice, checking the Jacobi identity is often the most difficult part of establishing a Courant algebroid structure. For this reason, it is useful---as is done in~\cite{RoytenbergCDA}---to define the notion of a \emph{almost Courant algebroid},\footnote{A more common name for this object in the literature is \emph{metric algebroid} \cite{Vaisman:2012ke}; we adopted the different name here in order to emphasize the relation of this object to a weakening of the Courant algebroid axioms.} which is the same as Definition~\ref{dfn:CA} with the condition~\eqref{it:Jacobi} dropped.

There are a pair of obstructions that measure the failure of an almost Courant algebroid to be a Courant algebroid. In this section, we establish sufficient conditions that ensure that these obstructions vanish (Lemma~\ref{lem:FV}). In the sequel, we will use this lemma in the proof of Theorem~\ref{thm:courant-extension}.

\begin{dfn}
Let $\cE$ be an almost Courant algebroid. The \emph{Jacobiator} is the map defined on triples of sections of~$\mc E$ by the formula
\[
\Jac\colon \cE \times \cE \times \cE \to \cE, \qquad (u,v,w) \mapsto [u,[v,w]]-[v,[u,w]]-[[u,v],w] . 
\] 
Similarly, there is a map $R$ (called the \emph{curvature} in~\cite{Moucka:2025kay}), defined on pairs of sections by the formula
\[
R\colon \cE \times \cE \to \Der \cO, \qquad (u,v) \mapsto \rho([u,v]) - [\rho(u), \rho(v)]. 
\]
\end{dfn}
It is clear that $R$ measures the failure of the anchor map to be a morphism of Leibniz algebras, whereas the Jacobiator measures the failure of the Jacobi identity: it vanishes precisely when property~\eqref{it:Jacobi} holds.
  
\begin{lem}
\label{lem:FV}
Let $\mc E$ be an almost Courant algebroid.
\begin{enumerate}
\item Suppose $\ker(\rho) \subset \mc E$ is coisotropic.
\label{it:FV1}
Then $R$ is antisymmetric and $\cO$-bilinear.
\end{enumerate}
Suppose that $R = 0$. Then the following statements also hold:
\begin{enumerate}[resume]
\item 
For any $f\in \cO$ and $u\in\mc E$ we have $[u, \cD f ] = \cD \rho(u) f$.
\label{it:FV2}
\item For any $f \in \cO$, the map $[\cD f, \cdot]$ is the zero map.
\label{it:FV3}
\item $\Jac$ is antisymmetric and $\cO$-trilinear.
\label{it:FV4}
\end{enumerate}
\begin{proof}
We prove the statements in turn. 
\begin{enumerate}
\item
We note that $R$ is $\cO$-linear in the second argument:
\[
R(u,fv) - fR(u,v)
= \rho([u,fv]) - [\rho(u), f\rho(v)] \\
= \rho(\rho(u)(f) \cdot v) - \rho(u)(f) \cdot \rho(v) \\
= 0.
\]
Symmetrizing $R$, we see that
\[
R(u,v) + R(v,u) = \rho( [u,v] + [v,u] ) = \rho (  \cD \langle u,v \rangle )=\rho\,\eta^{-1} \rho^\vee \ud\la u,v\ra=0,
\]
establishing the $\cO$-linearity of $R$ also in its first argument.
\item
We observe that 
\[
\langle [u, \cD f ], v \rangle = \rho(u) \langle \cD f, v \rangle - \langle \cD f, [u,v] \rangle 
= \rho(u) \rho(v) f - \rho([u,v]) f 
= \rho(v) \rho(u) f
= \langle \cD \rho(u) f , v\rangle.
\]
The first equality is pairing compatibility, the third is the vanishing of $R(u,v)$, and the second and fourth follow from the definition of~$\cD$.
\item To see this, we observe, using property~\eqref{it:failed-anti-symmetry} and the second statement of the lemma, that
\[
[\cD f, u] = -[u, \cD f] + \cD \langle \cD f, u \rangle = - \cD \rho(u) f + \cD \rho(u) f = 0.
\]
\item From the proof of Lemma~\ref{lem:Leibmap}, we see that the vanishing of~$R$ implies the linearity of~$\Jac$ in the third argument. So it suffices to show antisymmetry. On the first two arguments, we observe that
\[
\Jac(u,v,w) + \Jac(v,u,w) = \left[ \left( [u,v] + [v,u] \right), w \right] = [\cD \langle u, v \rangle, w ],
\]
which vanishes by the third statement of the lemma. Finally, on the second and third arguments, the result follows from a computation:
\[
\begin{aligned}
\Jac(u,v,w) &= [u,\cD\langle v,w \rangle] - [u,[w,v]] - \cD \langle [u,v]  ,w\rangle  + [w,[u,v]] - \cD \langle v, [u,w] \rangle + [[u,w],v] \\
&= \cD \rho(u)\langle v,w\rangle - \cD \langle [u,v]  ,w\rangle - \cD \langle [u,w]  ,v\rangle - \Jac(u,w,v) \\
&= - \Jac(u,w,v).
\end{aligned}
\]
The first equality 
 uses property~\eqref{it:failed-anti-symmetry}, 
 second equality uses the second statement of the lemma, and the third uses pairing compatibility (property~\eqref{it:pairing-compatibility}).\qedhere
\end{enumerate}
\end{proof}
\end{lem}
\begin{cor}\label{cor:weak-ca}
  A Courant algebroid is precisely an almost Courant algebroid for which $\ker\rho$ is coisotro\-pic and both $R$ and $\Jac$ vanish on some local basis of sections of~$\mc E$.
\end{cor}

\section{Two constructions}\label{sec:two-constructions}

\subsection{Reduction}\label{subsec:reduction}
Symplectic reduction has an analogue in the setting of Courant algebroids \cite{let}.
A notion of holomorphic reduction was studied in~\cite[Appendix~A]{GenKaehler}. We review the procedure here, setting it up in our language and at our level of generality.

Let $\mc E$ be a Courant algebroid over $\mc O$ and $i\colon \mc L\subset\mc E$ a sheaf of (locally semifree) involutive isotropic submodules.
The last two conditions mean that \[[\mc L,\mc L]\subset\mc L,\qquad \la \mc L,\mc L\ra=0.\]
It follows from isotropy that $\mc L$ is naturally a Lie algebroid over $\mc O$. Since
\[\la [\mc L,\mc L^\perp],\mc L\ra= \la \mc L^\perp,[\mc L,\mc L]\ra\subset \la \mc L^\perp,\mc L\ra=0,\]
we have $[\mc L,\mc L^\perp]\subset \mc L^\perp$, and 
the bracket descends to define a flat $\mc L$-connection on the sheaf $\mc F\coloneqq\mc L^\perp/\mc L$, given by the formula
\[
\nabla_l \pi(u) \coloneqq \pi ( [l,u] )  
\]
for $l\in\mc L$ and $u \in \cL^\perp$, with $\pi$ the natural quotient map.
We can then define the sheaves
\[\mc O^\mc L\coloneqq\{f\in \mc O\mid\rho(\mc L)f=0\},\qquad \mc F^\mc L\coloneqq\text{$\mc L$-flat sections of $\mc L^\perp/\mc L$},\]
and the latter is a sheaf of modules over the former.
\begin{thm}[Reduction of $\mc O$-Courant algebroids]\label{thm:reduction}
 Suppose $\mc F^\mc L\otimes_{\mc O^\mc L}\mc O \to \mc F$ is locally an isomorphism.
  % of $\cO$-modules. 
 Then the $\mc O$-Courant algebroid structure on $\mc E$ induces an $\mc O^\mc L$-Courant algebroid structure on $\mc F^\mc L$.
\end{thm}
\begin{proof}
Clearly, the anchor map on $\mc E$ descends to a well-defined map $\mc F^\mc L\to\Der\mc O^\mc L$. Next, we observe that the pairing $\eta$ descends to a nondegenerate pairing on $\cF$, via the following argument.
Consider the defining short exact sequence of $\cO$-modules
\[
\begin{tikzcd}
0 \ar[r]  & \cL^\perp \ar[r, "j"] &  \cE \ar[r, "i^*\circ\eta"] & \cL^\vee \ar[r] & 0.
\end{tikzcd}
\]
Dualizing this sequence and using the self-duality of $\eta$ gives us a sequence
\[
\begin{tikzcd}
0 \ar[r]  & \cL \ar[r, "i"] &  \cE \ar[r, "j^*\circ\eta"] & (\cL^\perp)^\vee \ar[r] & 0,
\end{tikzcd}
\]
which is identified with the defining sequence of~$(\cL^\perp)^\perp$. Thus $\cL = (\cL^\perp)^\perp$. (Note that we have used that $(\cL^\vee)^\vee = \cL$, which is where we use the assumptions on~$\mc L$; the result would hold for $\cL$ finitely generated and reflexive.)

By isotropy, $\cL \subset \cL^\perp$. 
We thus obtain a short exact sequence of two-term cochain complexes:
\[
\begin{tikzcd}
0 \ar[r]  & \cL \ar[r, "i"] \ar[d]  &  \cE \ar[r, "j^*\circ\eta"] \ar[d, "\cong"]  & (\cL^\perp)^\vee \ar[d]   \ar[r] & 0 \\
0 \ar[r]  & \cL^\perp \ar[r, "j"] &  \cE \ar[r, "i^*\circ\eta"] & \cL^\vee \ar[r] & 0.
\end{tikzcd}
\]
This gives rise to a long exact sequence in cohomology whose first connecting map is a map from $\cF^\vee$ to~$\cF$, necessarily an isomorphism by exactness. Thus the induced pairing on $\mc F$ is nondegenerate.
    
    By the invariance of the inner product this descends to a map (which we will also decorate with a subscript $\mc L$)
    \[\la\slot,\slot\ra_\mc L\colon\mc F^\mc L\otimes_{\mc O^\mc L}\mc F^\mc L\to\mc O^\mc L.\]
    Next, we note that by assumption we have that for any $u\in \mc F^\mc L$
    \[\la u,\mc F^\mc L\ra_\mc L=0\quad\implies \quad \la u,\mc F\ra_\mc L=0 \quad\implies\quad u=0,\]
    showing $\ker\eta_\mc L=0$, and by self-duality also $\coker\eta_\mc L=0$, implying the non-degeneracy on $\mc F^\mc L$.
    
    Take $u,v\in\mc L^\perp$ such that their image in $\mc F$ is flat, i.e.\ $[\mc L,u]\subset\mc L$ and $[\mc L,v]\subset\mc L$.
    Then
    $\la [u,v],\mc L\ra=\la v,[u,\mc L]\ra=\la v,[\mc L,u]\ra=0$ and hence $[u,v]\in\mc L^\perp$.
    Furthermore,
    \[[\mc L,[u,v]]\subset[[\mc L,u],v]+[u,[\mc L,v]]\subset [\mc L,v]+[u,\mc L]\subset \mc L,\]
    implying the image of $[u,v]$ in $\mc F$ is flat. Finally, we need to check that the result does not depend on the particular representatives $u,v\in\mc L^\perp$, i.e.\ we need to check that $[u,\mc L],[\mc L,v]\subset\mc L$, which is again immediate. Thus the bracket on $\mc F^\mc L$ is well-defined. Checking the Courant algebroid axioms is then straightforward.
  \end{proof}
  
    \begin{eg} \label{ex:holo}
  Let $\mc E=\Gamma(TM \oplus T^*M)$ be the standard smooth Courant algebroid over a complex manifold $M$, with $\mc L=\Gamma(T^{0,1}M)$. Then $\mc O^\mc L=\mc C_{\hol}$ is the sheaf of holomorphic functions and $\mc F^\mc L$ is the standard holomorphic Courant algebroid.
  \end{eg}
  \begin{eg}
    Suppose $\mc E$ is a Courant algebroid over $\mc O$ and $\mc L$ is involutive and Lagrangian, i.e.\ $\mc L^\perp=\mc L$. Then $\mc O^\mc L$ is the sheaf of $\rho(\mc L)$-invariant functions, and $\mc F^\mc L=0$. For instance, we can take the standard smooth Courant algebroid $\mc E=\Gamma(TM\oplus T^*M)$ over a smooth manifold $M$, with $\mc L=\Gamma(TM)$ leading to $\mc O^\mc L=\underline{\mb R}$, or $\mc L=\Gamma(T^*M)$ leading to $\mc O^\mc L =C^\infty$.
  \end{eg}

\subsection{Extension of scalars}
Let $i\colon \cO \to \cO'$ be a map of (sheaves of) cdgas. Recall that we have 
a corresponding forgetful functor,
\[
\Res_i\colon \OPrimeMod \to \OMod, 
\]
called \emph{restriction of scalars}. The functor has a left adjoint, 
\[
\Ind_i\colon \OMod \to \OPrimeMod,
\] called \emph{extension of scalars;} the adjunction is witnessed by the fact that 
\deq[eq:adjunction]{
\Hom_\cO(X, \Res_i Y) \cong \Hom_{\cO'}(\Ind_i X, Y)
}
for any $\cO$-module $X$ and any $\cO'$-module $Y$.

Since Courant algebroids are defined by differential operators, there is no \emph{a priori} way to apply extension of scalars to a Courant algebroid over $\cO$ and obtain a Courant algebroid over $\cO'$. In this section, we will show that choosing one additional datum (roughly, a lift of the anchor map to $\Der \cO'$ satisfying certain conditions) allows one to construct a Courant algebroid structure on
\[\cE' \coloneqq \cE \otimes_\cO \cO'.\]
This structure is the unique Courant algebroid structure over $\cO'$ extending the one on~$\cE$ and having the specified anchor map.\footnote{Our assumptions on the structure sheaf of a smooth dg ringed manifold mean that we do not have to worry about the derived tensor product here.}

Our result is essentially a common generalization of results of Gr\"utzmann~\cite{Gruetzmann} for holomorphic Courant algebroids
and of Li-Bland--Meinrenken~\cite{li2009courant}
 for extensions of the bracket on locally constant functions on~$M$ valued in a quadratic Lie algebra $\fg$ to a Courant algebroid structure on all smooth $\fg$-valued functions.
More generally, one can see them as an adaptation of the construction of the unique extension of holomorphic differential operators to the Dolbeault complex, and thus to ($\dbar$-compatible) differential operators on all smooth functions. In the holomorphic setting, there is a natural choice of lifted anchor map, as we will explain later; we first pin down the extra data required in general.

Recall that $i$ does not induce any natural map from $\Der \cO$ to $\Der \cO'$. Rather, there is a pair of maps involving the $\cO'$-module $\Der(\cO, \cO')$. We have the following maps:
Restriction along $i$ yields a map
\[
i^* \in \Hom_{\cO'}\left(\Der \cO', \Der(\cO, \cO')\right).
\]
On the other hand, we have natural maps of $\cO$-modules
\[
\begin{tikzcd}
\Der(\cO) \ar[r, "1 \otimes i"] & \Der(\cO) \otimes_\cO \cO' \ar[r, "\cong"] & \Der(\cO, \cO').
\end{tikzcd}
\]
To gain intuition for these maps, it is illustrative to imagine $\cO$ to be the sheaf of holomorphic functions on a complex manifold $M$ and $\cO'$ the sheaf of complex-valued smooth functions. Then $\Der(\cO)$ is identified with holomorphic vector fields, which are holomorphic sections of $(TM \otimes \C)/T^{0,1}M$, and $\Der(\cO')$ is identified with smooth sections of $TM \otimes \C$. The intermediate object $\Der(\cO,\cO')$ then corresponds to smooth sections of $(TM \otimes \C)/T^{0,1}M$; the natural lift identifies these with smooth sections of the holomorphic tangent bundle $T^{1,0}M$.

More generally, we capture the necessary datum with the following definition.
\begin{dfn} 
\label{dfn:hrho}
Let $\cE$ be a Courant algebroid over $\cO$ and $i\colon \cO \to \cO'$ as above. A \emph{lift of the anchor map} is a map of $\cO$-modules
\[
\hrho\colon \cE \to \Res_i \Der \cO'
\]
satisfying the following three properties:
\begin{enumerate}
\item $\hrho$ preserves brackets, in the sense that
\label{it:hro-bracket}
\[
\hrho([u,v])= \hrho(u)\hrho(v) - \hrho(v) \hrho(u);
\]
\item The extension of $\hrho$ to an $\O'$-linear map $\rho'\colon \cE' \to \Der \cO'$---using the adjunction~\eqref{eq:adjunction}---extends the original anchor map, in the sense that the diagram
\begin{equation}\label{lift-rho-cd}
\begin{tikzcd}
\cE' \ar[r, "\rho'"] \ar[d, "\Ind_i \rho"'] & \Der\cO' \ar[d,"i^*"] \\
\Der \cO \otimes_\cO \cO' \ar[r] & \Der(\cO,\cO')
\end{tikzcd}
\end{equation}
commutes.
\label{it:hrho-lift}
\item
$\ker(\rho')$ is coisotropic.
\label{it:hrho-coiso}
\end{enumerate}
\end{dfn}
\begin{thm}[Extension of scalars for Courant algebroids]
\label{thm:courant-extension}
For any Courant algebroid $\cE$ and any lift $\hrho$ of the anchor map, there exists a unique Courant algebroid structure on~$\cE'$ over $\cO'$ with corresponding anchor map $\rho'$ that restricts to the given Courant algebroid structure on~$\cE$.
\begin{proof}

Any Courant algebroid structure is uniquely determined by prescribing the pairing, anchor, and bracket on a basis of the space of sections. Since in our case the sections of $\mc E$ span $\mc E'$ over $\mc O'$, the uniqueness follows.

We define the operations on $\cE'$ by taking the anchor to be $\rho'$ and the pairing to be $\eta' = \Ind_i \eta$. 
The bracket is defined (\emph{a priori} on pairs elements of $\cE \otimes_\kk \cO'$) by the formula (cf.\ \cite[p.~49]{Gruetzmann})
\[
[u \otimes_\kk a, v \otimes_\kk b]' \coloneqq [u,v] \otimes_\cO ab + v \otimes_\cO a\hrho(u)b - u \otimes_\cO b \hrho(v)a + b \langle u, v\rangle \cD' a \in \cE'.
\]
We now check the axioms of a Courant algebroid, as well as the fact that the bracket descends to $\cE'$, in a convenient order.
\begin{itemize}[leftmargin=11pt]
\item \emph{Failed antisymmetry:} We symmetrize the bracket and compute that 
\[
\begin{aligned}
[u \otimes_\kk a, v \otimes_\kk b]' + [v \otimes_\kk b, u \otimes_\kk a]' &= \left( [u,v] + [v,u] \right) \otimes_\cO ab + b \langle u, v\rangle \cD' a + a \langle u, v\rangle \cD' b \\
&= ab  \cD(\langle u, v \rangle ) + (\cD' a) b \langle u, v\rangle  + a(\cD' b) \langle u, v\rangle  \\
&= \cD' \left( ab \langle u, v \rangle \right) = \cD' \langle u \otimes_\kk a, v \otimes_\kk b \rangle'.
\end{aligned}
\]
The second equality uses failed antisymmetry for the bracket on~$\cE$, the third uses property~\eqref{it:hrho-lift} in Definition~\ref{dfn:hrho}, and the last uses the definition of~$\eta'$.
\item \emph{Descent to~$\cE'$:} We need to show that the bracket vanishes on all elements of the form
\[
v \otimes_\kk zb - zv \otimes_\kk b,
\]
where $z \in \cO$. On the second argument, we see that
\[
\begin{aligned}
[u \otimes_\kk a, v \otimes_\kk zb]'  -  [u \otimes_\kk a, zv \otimes_\kk b]'  
&= [u,v] \otimes_\kk zab - [u, zv] \otimes_k ab + v \otimes_\cO a\hrho(u)(zb) - zv \otimes_\cO a\hrho(u)b \\
&= - v \otimes \rho(u)z \cdot ab + v \otimes \hrho(u)z \cdot ab,
\end{aligned}
\]
which vanishes because of condition~\eqref{it:hrho-lift} on~$\hrho$. Since failed antisymmetry holds and the third term in the relation clearly vanishes on the ideal, we deduce that the bracket vanishes on the ideal in the first argument as well. We thus no longer need to write $\otimes_\kk$; in the remainder of the proof, $\otimes = \otimes_\cO$ for brevity.
\item \emph{Leibniz:} We again check by computation that 
\[
[u \otimes a, v \otimes fb]'  - f [u \otimes a, v \otimes b]'  
= v \otimes ab \hrho(u)f = \rho'(u\otimes a)f \cdot (v \otimes b).
\]
\item \emph{Pairing compatibility:}
Finally, we note that
\begin{equation*}
\begin{aligned}
	\rho'(u \otimes a) \langle v\otimes b, w \otimes c \rangle
		&= abc\,  \rho(u)\langle v,w \rangle 
			+ a \langle v,w \rangle ( c\hrho(u)b  + b\hrho(u)c ) \\
		&= \langle [u, v] \otimes ab + v \otimes a\hrho(u)b , w \otimes c\rangle 
		+ \langle v, [u,w]\otimes ac + w\otimes a\hrho(u)c   \rangle \\
		&= \langle [{u\otimes a} , v\otimes b]', w\otimes c \rangle 
			+ \langle v\otimes b,  [{u\otimes a} , w\otimes c]' \rangle \\
		& \qquad+ \left\langle u \otimes b \hrho(v)a - \langle u,v \rangle b\cdot \cD' a, w \otimes c \right\rangle
			+ \left\langle v \otimes b, u \otimes c \hrho(w)a - \langle u,w \rangle c\, \cD' a \right\rangle.
\end{aligned}
\end{equation*}
We see that the first line of the last equality expresses pairing compatibility. The extra terms on the second line can be rewritten as
\begin{equation*}	 
	\langle u, v \rangle b \Big( c\hrho(w)a - \langle w \otimes c, \cD' a \rangle \Big)  
			+ \langle u, w \rangle c \Big( b\hrho(v)a - \langle v \otimes b, \cD' a \rangle\Big),
\end{equation*}	
and thus vanish after recalling the definition of~$\cD'$. 
\end{itemize}
We have thus shown that the structure on $\mc E'$ defines an almost Courant algebroid in the sense of~\S\ref{ssec:weak}.
Using~\eqref{it:hrho-coiso} and the fact that both $\Jac$ and $R$ vanish on sections of $\mc E$, which generate $\mc E'$ over $\mc O'$, Corollary~\ref{cor:weak-ca} allows us to conclude that $\mc E'$ is a Courant algebroid.
\end{proof}
\end{thm}

\begin{eg}[Courant algebroids from infinitesimal group actions~\cite{li2009courant}]
  Let $M$ be a smooth manifold, $(\mf g, \inner{\ , \ }_{\mf g})$ a quadratic Lie algebra and $\cO = \underline{\R}$ the sheaf of locally constant functions on $M$. Then $\mc E = \underline{\mf g}$ defines an $\cO$-Courant algebroid with the anchor $\rho = 0$ and the bracket and inner product on $\mc E$ given by the bracket and inner product on $\mf g$.

  Choosing $\cO' = C^\infty$ with $i$ the natural map from locally constant to smooth functions, we find that $\mc E' = \ul{\mf g} \otimes_{\underline\R} C^\infty$ consists of smooth functions valued in~$\fg$. Since $\Der C^\infty = \mf X$ is the sheaf of smooth vector fields, to specify a lift of $\rho$ is to specify a Lie map $\hrho \colon \mf g \to \mf X$, i.e.\ an action of $\mf g$ on $M$. Note that the adjunction in \eqref{eq:adjunction} is represented by the fact that we can interpret $\hrho$ either as the $\R$-linear map $\mf g \to \mf X$ above or equivalently as a $C^\infty$-linear map on smooth sections, i.e.\ a bundle map $\mf g \times M \to TM$.

  Following Theorem~\ref{thm:courant-extension}, the inner product on $\mc E'$ is then given by the extension and the bracket is
  \begin{equation*}
    \begin{split}
      [x_1 \otimes f_1, x_2 \otimes f_2]' &= [x_1, x_2]_{\mf g} \otimes f_1 f_2 + x_2 \otimes f_1 \hrho(x_1)f_2 - x_1 \otimes f_2 \hrho(x_2)f_1 + f_2 \inner{x_1, x_2}_{\mf g} \mc D' f_1\\
    \end{split}
  \end{equation*}
  The requirement that $\ker \rho'$ is coisotropic in $\mc E'$ then reduces to the statement that the stabilizer of any point $p \in M$, i.e.\ $\ker\hat \rho_p = \{x\in \mf g \mid \hat \rho (x)_p = 0 \}$, is coisotropic. This is precisely the condition identified in~\cite{li2009courant}.
\end{eg}

In the remainder of the paper we will be interested in a particular class of lifts $\hat \rho$, which arise from a construction involving the structure sheaves $\mc O$ and $\mc O'$ directly, without referencing any particular Courant algebroid structure.

\begin{prop}\label{prop:postcompose}
  Fix a cdga map $i\colon \cO \to \cO'$ and let $\lambda\colon \Der\mc O\to \Res_i\Der \mc O'$ be a bracket preserving map such that the following commutes:
  \begin{equation}\label{new-commutative}
  \begin{tikzcd}
    \Der\mc O \ar[r,"\lambda"] \ar[dr,"i"'] & \Res_i\Der\mc O'\ar[d,"i^*"]\\
    & \Res_i\Der(\mc O,\mc O').
  \end{tikzcd}
  \end{equation}
  Then for any Courant algebroid $\mc E$ over $\mc O$, the composition $\hat\rho\coloneqq\lambda\circ \rho$ is a lift of the anchor map (cf.\ \ref{dfn:hrho}).
\end{prop}
\begin{proof}
  Since $\hat\rho$ is a composition of two bracket-preserving maps, it also preserves the brackets. 
  Next, we note that $\Ind_i\lambda$ defines a diagonal map in \eqref{lift-rho-cd} from the bottom left to the top right corner. The commutativity of the upper and lower triangle corresponds to the definition $\hat\rho\coloneqq\lambda\circ \rho$ and the commutativity of \eqref{new-commutative}, respectively. Finally, the coisotropy follows from
  \[\hat\rho \, \eta^{-1} \hat\rho^\vee=\lambda (\rho \, \eta^{-1} \rho^\vee) \lambda^\vee=0.\qedhere\]
\end{proof}

\begin{eg}[Holomorphic to smooth]
Let $\cO$ be the sheaf of holomorphic functions on a complex manifold, $\cO'$ the sheaf of complex-valued smooth functions, and $i\colon \mc O\to\mc O'$ the obvious inclusion. Taking $\lambda$ to be the natural embedding of holomorphic vector fields into smooth ones, one recovers the procedure \cite{Gruetzmann} of extending a holomorphic Courant algebroid to a smooth complex-valued one.
\end{eg}

Another interesting example is obtained by taking $i$ to be the map from holomorphic functions to the Dolbeault complex $(\Omega^{0,\bullet},\bar\partial)$. This will be examined in more detail in \S\ref{subsec:red-ext} below.

\section{The Courant contact model}
\label{sec:five}

\subsection{The Roytenberg--Weinstein local Lie algebra}
As the reader may have observed,  condition~\eqref{it:failed-anti-symmetry} in the definition of a Courant algebroid expresses the fact that the symmetrization of the bracket operation is in the image of the operator $\cD$, in a manner that is witnessed by the pairing. The failure of the bracket to be a Lie bracket on~$\cE$ therefore becomes exact in the cochain complex
\begin{equation*}
 \Cone(\cD) = \begin{tikzcd} \cO[1] \ar[r, "\cD"] & \cE. \end{tikzcd}
\end{equation*}
One is thus led to imagine using the Courant algebroid structure on~$\cE$ to define a local $L_\infty$ algebra on~$\Cone(\cD)$.
Exactly this was done by Roytenberg and Weinstein~\cite{RW}.
To do so, one antisymmetrizes the bracket, thus trading the failure of antisymmetry for a failure of the Jacobi identity. 
This failure is again exact, and can be corrected by introducing a single 3-ary higher bracket operation.

\begin{dfn}
Let $\cE$ be a Courant algebroid over~$\cO$. The \emph{Roytenberg--Weinstein local Lie algebra} of~$\cE$ is the locally free sheaf of dg $\cO$-modules
\[
\RW(\cE) \coloneqq \Cone(\cD),
\]
viewed as a local Lie algebra with the following operations:
\begin{itemize}
    \item The differential is given by the differential $\cD + \dbar$ on the dg $\O$-module.
    \item The 2-ary bracket on pairs of sections in degree 0 is the symmetrization of the Courant bracket:
    \begin{equation*}
  \mu_2(\xi_1,\xi_2) \coloneqq \tfrac12 ([\xi_1,\xi_2]-[\xi_2,\xi_1]) = [\xi_1,\xi_2] - \tfrac{1}{2} \cD \eta(\xi_1,\xi_2).
\end{equation*}
    \item The 2-ary bracket of a section $u \in \cE$ with a section $f \in \cO[1]$ is given by
    \begin{equation*}
    \mu_2(\xi,f) = \tfrac{1}{2} \langle \xi,\cD f \rangle = \tfrac{1}{2} \rho(\xi)f.
    \end{equation*}
    \item There is a single $3$-ary bracket correcting for the failure of the Jacobi identity, which is
\begin{equation*}
    \mu_3(\xi_1,\xi_2,\xi_3) = - \tfrac{1}{6}\left( \langle [\xi_1,\xi_2],\xi_3 \rangle + \langle [\xi_2,\xi_3],\xi_1 \rangle + \langle [\xi_3,\xi_1],\xi_2 \rangle \right) \in \cO[1].
\end{equation*}
\end{itemize}
Here we denote sections of~$\cE$ with the letter $\xi$ and section of~$\cO[1]$ with the letter $f$; we will maintain this notation when referring to the fields of our model below.
\end{dfn}

The anchor map on~$\cE$ becomes a map of local Lie algebras from $\RW(\cE)$ to $\Der\cO$, so that $\RW(\cE)$ is naturally a higher Lie algebroid. Indeed, the Roytenberg--Weinstein local Lie algebra can be thought of as the Atiyah algebroid of a gerbe with connective structure. (See~\cite[Letter~5]{let}, and also the recent discussion in~\cite{BunkShahbazi}.)
As recalled above in~\S\ref{ssec:locLie}, $\RW(\cE)$ presents a sheaf of formal moduli problems, and we think of the cosheaf $C^\bu(\RW(\cE))$ as functions on that sheaf of spaces.
We conclude this discussion by noting that the Chevalley--Eilenberg differential (or homological vector field) can be explictly written as 
\begin{equation}
Q_{\RW} = \int_M\left(\frac12\ms L_\xi f-\frac1{12}\la \xi,\ms L_\xi\xi\ra\right)\frac{\delta}{\delta f}+\left(\mc Df+\frac12 \ms L_\xi\xi\right)\frac{\delta}{\delta \xi}.
\label{eq:QRW}
\end{equation}

\subsection{The master action}
Suppose now that a Courant algebroid $\mc E$ over $\mc O$ is equipped with a fixed $n$-orientation (in the sense of Definition~\ref{dfn:or}), given by a section \[\lambda_0\in\mc O^![-n].\] For brevity, we say $\lambda_0$ is \emph{nondegenerate}. We then introduce the following sheaf of graded locally free $\O$-modules: 
\begin{equation}
  \ms M[1] = \O[2] \oplus \cE[1] \oplus \O^![-n]
\label{eq:M}
\end{equation}
and regard it as the space of fields of the Courant contact model, although the full differential contains additional terms; look ahead to Remark~\ref{rmk:semifree}. We will equip $\ms M$ with the structure of a local Lie algebra below. We denote the fields by $(f,\xi,\lambda')$, extending our conventions for $\RW(\cE)$. 

The field $\lambda'$ can be interpreted as a deformation of the background orientation $\lambda_0$; as such, we define $\lambda = \lambda_0 + \lambda'$, and write all formulas using the variable $\lambda$.
There is furthermore a natural action of the group $\cO^\times$ of invertible functions on~$\ms M[1]$, by rescaling the orientation $\lambda$; this action does not preserve $\lambda_0$, but it does carry nondegenerate orientations to nondegenerate orientations. 

Define the local one-form
\deq[eq:theta]{
    \vartheta \coloneqq \int \lambda \left( \delta f + \frac12 \langle \xi, \delta \xi \rangle \right)
}
on $B \ms M$. Note that this form is linear with respect to the field space coordinate $\lambda$.
Furthermore, viewing $\vartheta$ as a Liouville one-form defines an exact $(2-n)$-shifted symplectic structure on~$B \ms M$, given by the local closed two-form
\deq[eq:omega]{
    \omega \coloneqq \delta\vartheta = \int \delta \lambda \left( \delta f + \frac12 \langle \xi, \delta \xi \rangle \right) + \frac{(-1)^n}2 \lambda \langle \delta \xi, \delta \xi \rangle
} 
Note that this symplectic form is \emph{not} written in Darboux form. 

\begin{lem}
  The Hamiltonian vector field for the degree $3-n$ functional
  \[S_{\bar\partial}=\int_M\lambda\left(\bar\partial f+\frac12\la \xi,\bar\partial \xi\ra\right)\]
  is the internal differential $\dbar$.
\end{lem}
\begin{proof}
  Since $\bar\partial$ preserves all the structure involved, we have
  \[i_{{\bar\partial}}\omega=i_{{\bar\partial}}\delta\vartheta=-L_{{\bar\partial}}\vartheta+\delta i_{{\bar\partial}}\vartheta=\delta i_{{\bar\partial}}\vartheta=-\delta\int_M\lambda\left(\bar\partial f+\frac12\la\xi,\bar\partial\xi\ra\right).\qedhere\]
\end{proof}

\begin{thm}\label{thm:cme}
The degree $3-n$ action functional
$S=S_0+S_{\bar\partial}$, with
  \[S_0=\int \lambda \left( \ms L_\xi f + \frac16 \langle \xi,\ms  L_\xi \xi \rangle\right),\]
satisfies the classical master equation on $(\ms M[1], \omega)$.
\end{thm}
Thus $S$ defines a local Lie algebra structure on~$\ms M$. The corresponding formal moduli problem is the Courant contact model.
\begin{proof}
  Since $\bar\partial^2=0$, we have $\{S_{\bar\partial},S_{\bar\partial}\}=0$. Furthermore, since $\bar\partial$ preserves the Courant algebroid structure, we also have $\{S_0,S_{\bar\partial}\}=0$. It thus remains to check $\{S_0,S_0\}=0$.
  
  First, under the variation $\xi\to\xi+\delta\xi$ we get
  \begin{equation*}
  \begin{split}
    \delta\la \ms L_\xi\xi,\xi\ra&=\la \ms L_\xi\delta\xi+\ms L_{\delta\xi}\xi,\xi\ra+\la \ms L_\xi\xi,\delta\xi\ra=\la 2\ms L_\xi\delta\xi+\cD\la \delta\xi,\xi\ra,\xi\ra+\la \ms L_\xi\xi,\delta\xi\ra\\
    &=2\la \ms L_\xi\delta\xi,\xi\ra+\ms L_\xi\la \delta\xi,\xi\ra+\la \ms L_\xi\xi,\delta\xi\ra=3\la \ms L_\xi\delta\xi,\xi\ra,
  \end{split}
  \end{equation*}
  and under the same variation, we find
  \[\delta S_0=\int_M\lambda(\ms L_{\delta\xi}f+\tfrac12\la \ms L_\xi\delta\xi,\xi\ra)=\int_S\la\delta\xi,(-1)^n\lambda\cD\!f+\tfrac12\ms L_\xi(\lambda\xi)\ra.\]
  Thus for the variational derivatives we have
  \[\frac{\delta S_0}{\delta\xi}=(-1)^n\lambda\cD\!f+\frac12\ms L_\xi(\lambda\xi),\qquad \frac{\delta S_0}{\delta f}=(-1)^{n+1}\ms L_\xi\lambda,\qquad \frac{\delta S_0}{\delta \lambda}=\ms L_\xi f + \frac16 \langle \xi, \ms L_\xi \xi \rangle.\]
  Solving $i_{X_{S_0}}\omega=-\delta S_0$, we obtain the homological vector field
  \deq[eq:QBV]{
  X_{S_0}=\int_M\left(\frac12\ms L_\xi f-\frac1{12}\la \xi,\ms L_\xi\xi\ra\right)\frac{\delta}{\delta f}+\left(\mc Df+\frac12 \ms L_\xi\xi\right)\frac{\delta}{\delta \xi}+(\ms L_\xi\lambda) \frac{\delta}{\delta \lambda}.
  }
  Note that this consists just of the homological vector field for the Roytenberg--Weinstein local Lie algebra, as given above in~\eqref{eq:QRW}, together with a term representing the action by Lie derivative on the space of volume forms.
  
  Plugging this into the classical master equation and using
  \[\la \cD\!f,\cD\!f\ra=0,\qquad \ms L_{\ms L_\xi\xi}=2\ms L_\xi \ms L_\xi,\]
  we obtain
  \begin{align*}
    \{S_0,S_0\}&=X_{S_0}S_0=\int_M\frac12(\ms L_\xi\lambda)\ms L_\xi f-\frac1{12}(\ms L_\xi \lambda)\la \xi,\ms L_\xi\xi\ra+\frac12\la \ms L_\xi(\lambda\xi),\mc Df\ra+\frac12(-1)^n\la \lambda \mc Df,\ms L_\xi\xi\ra\\
    &\qquad\qquad\qquad\qquad\qquad +\frac14\la \ms L_\xi(\lambda\xi),\ms L_\xi\xi\ra+(\ms L_\xi \lambda)\ms L_\xi f+\frac16(\ms L_\xi\lambda)\la \xi,\ms L_\xi\xi\ra\\
    &=\int_M2(\ms L_\xi \lambda)\ms L_\xi f+(-1)^n\lambda \ms L_{\ms L_\xi\xi}f+\frac13(\ms L_\xi\lambda)\la \xi,\ms L_\xi\xi\ra+\frac14(-1)^n\lambda\la \ms L_\xi\xi,\ms L_\xi\xi\ra\\
    &=(-1)^n\int_M\lambda\left(-\frac13\ms L_\xi\la \xi,\ms L_\xi\xi\ra+\frac14\la \ms L_\xi\xi,\ms L_\xi\xi\ra\right)=(-1)^n\int_M\lambda\left(-\frac1{12}\la \ms L_\xi\xi,\ms L_\xi\xi\ra+\frac13\la \xi,\ms L_\xi \ms L_\xi\xi\ra\right).
  \end{align*}
  To conclude, we note that
  \[\la \ms L_\xi\xi,\ms L_\xi\xi\ra-\la \xi,\ms L_\xi \ms L_\xi\xi\ra=\ms L_\xi\la \ms L_\xi\xi,\xi\ra=\la \xi,\cD\la \ms L_\xi\xi,\xi\ra\ra=\la \xi,\ms L_{\ms L_\xi\xi}\xi+\ms L_\xi \ms L_\xi\xi\ra=3\la \xi,\ms L_\xi \ms L_\xi\xi\ra,\]
  implying
  \[\la \ms L_\xi\xi,\ms L_\xi\xi\ra=4\la \xi,\ms L_\xi \ms L_\xi\xi\ra,\]
  and hence also the vanishing of $\{S_0,S_0\}$.
\end{proof}
\begin{rmk}
\label{rmk:semifree}
Note that the cochain complex of fields of the Courant contact model is not simply~\eqref{eq:M} equipped with $\dbar$. Rather, there is an additional differential, giving~$\ms M[1]$ the structure of a semifree $\cO$-module. By inspecting the linear part of the cohomological vector field~\eqref{eq:QBV}, we see that the additional terms in the differential are given by the operator $\cD$ (mapping $\cO[2]$ to $\cE[1]$), and by the map sending $\xi$ to $\ms L_\xi \lambda_0$  (mapping $\cE[1]$ to $\cO^![-n]$).
\end{rmk}
\subsection{Interpretation: contact structures and moduli spaces}
As we have seen in the proof, the local Lie algebra $\ms M$ is the semidirect product of $\RW(\cE)$ with the sheaf $\cO^![-n-1]$, with respect to the module structure given by Lie derivative. As such, it
fits into a short exact sequence
\[
0 \to \cO^![-n-1] \to \ms M \to \RW(\cE) \to 0.
\]
We interpret this as a fibre bundle over $B \RW(\cE)$, with fibre isomorphic to the sheaf of orientations on~$(M,\cO)$. 
This confirms the moduli-theoretic interpretation from the slogan given in~\S\ref{sec:intro}: the Courant contact model $B \ms M$ describes the moduli space of $n$-orientations considered up to generalized diffeomorphisms, working formally near a nondegenerate orientation. If we begin, following~\cite{let}, with a Courant algebroid that is the Atiyah algebroid of a gerbe with connective structure, we can interpret $B \ms M$ as a moduli space of pairs, consisting of a gerbe with connective structure and an orientation, considered just up to structure-preserving diffeomorphisms of~$(M,\cO)$.

Indeed, our computation above implicitly works globally in the fibre---but formally in the base---of this fibration. There is an obvious global notion of a nondegenerate orientation $\lambda$, and the sheaf of such is a torsor over $\cO^\times$. We can think of this as being the correct global object underlying the abelian formal moduli problem $B \cO^![-n-1]$. The correct global replacement for the base would be more subtle, since $\RW(\cE)$ is nonabelian; we comment on this in the context of generalized $G$-structures below.

Inspection makes it clear that this fits into a local version of the theory of shifted contact structures, defined by any of the local contact forms on~$B \RW(\cE)$ obtained from $\vartheta$ by setting $\lambda$ to a specific nondegenerate value. Following~\cite{Berktav1}, we can imagine presenting this structure by a subbundle of the tangent bundle of~$B\ms M$, mirroring the classical definition in terms of hyperplane distributions. The tangent bundle is presented by the local dg Lie algebra $C^\bu(\ms M, \ms M[1])$. With respect to an explicit basis for~$\cE$, we can consider this subbundle as presented by the map
\[
C^\bu(\ms M) \otimes \cE[1] \to C^\bu(\ms M, \ms M[1]), \qquad
e_\alpha \mapsto  \frac{\delta}{\delta \xi^\alpha} - \frac12 \eta_{\alpha \beta} \xi^\beta \frac{\delta}{\delta f}.
\]
Recall now that there is a classical equivalence between contact structures on a manifold $X$ and tuples consisting of a principal $\R^\times$-bundle $C$ over~$X$ and an exact symplectic form $\omega = d \vartheta$ on~$\Tot(C)$ for which $\vartheta$ is linear in the fibre coordinate. Under this equivalence, $(\Tot(C),\omega)$ is called the \emph{symplectification} of the contact manifold $X$. It is clear that our construction should fit into a local version of this story, and that the Courant contact model can be interpreted as the symplectification of the 
contact local moduli problem $B\RW(\cE)$.

\section{Twisted supergravity backgrounds} \label{sec:twisted-sugra-bg}
\subsection{$\ms N=1$ supergravity in ten dimensions}
Building upon \cite{bvsugra,Kupka:2025ulg}, let us now describe the data defining a twisted supergravity background for  $D=10$, $\ms N=1$ supergravity coupled to Yang--Mills multiplets. We will focus on the Euclidean version.\footnote{In particular, this means we take the fermion fields to be complex.} In this section, we fix $\mc E$ to be a smooth transitive Courant algebroid, corresponding to a vector bundle $E\to M$, with $\dim M=10$. Note that this implies that the signature of $E$ is $(10+p,10+q)$, with $p,q\ge0$. We will write \[H\to M\] for the line bundle of half-densities on $M$. We start by defining the bosonic field content of the theory, consisting of a spin structure and a dilaton.
\begin{dfn}
  A \emph{spin structure} is a reduction of the structure group of $E$ to \[\on{Spin}(10)\times O(p,10+q)\] along the map
  \[\on{Spin}(10)\times O(p,10+q)\to SO(10)\times O(p,10+q)\subset O(10+p,10+q).\]
  In particular, a spin structure defines an orthogonal splitting $E=C_+\oplus C_-$, with the signatures of $C_\pm$ being $(10,0)$ and $(p,10+q)$, respectively. The subbundle $C_+$ is called the \emph{generalized metric}. 
  Let us denote the (complex) chiral spinor bundles associated to $\on{Spin}(C_+)$ by $S_\pm$, and set $S=S_+\oplus S_-$. \end{dfn}
  
  In the context of this paper, we will in addition require that
  \[\rho|_{C_+}\colon C_+\to TM\]
  is an isomorphism.
From this, it follows that $\rho|_{C_-}\colon C_-\to TM\otimes \mb C$ is surjective (see e.g.\ the explicit form of $C_-$ in \cite{fullsugra}).

\begin{dfn}
  A \emph{dilaton} is an everywhere nonvanishing section
  $\sigma\in\Gamma(H)$.
\end{dfn}

Given a spin structure and a dilaton, there is a notion of 
\emph{generalized Levi-Civita connection}~\cite{CSW1,Garcia-Fernandez:2016ofz}.
Such connections exist, but are not unique;  we denote the set of all of them by~$\on{LC}$.
Importantly, for any $v\in \Gamma(C_-)$ and $\rho\in\Gamma(S\otimes H)$ the expressions $D_v\rho$ and $\di\rho$ are independent of the choice of the representative $D\in \on{LC}$ and hence depend only on the spin structure and the dilaton (in addition to $v$ and $\rho$). Similarly, denoting the projections onto $C_\pm$ by subscripts $\pm$, for any generalized Levi-Civita connection we have $D_{u_\pm}v_\mp=[u_\pm,v_\mp]_\mp$ for any $u,v\in\Gamma(E)$.

\begin{dfn}\label{dfn:twist-datum}
Given a spin structure and a dilaton, 
  a \emph{twist datum} for those data consists of an everywhere non-vanishing spinor $e\in\Gamma(S_+\otimes H)$ satisfying
  \begin{equation}\label{cond}
  \bar e\gamma_ae=0,\qquad \di e=0,\qquad D_ve=0,\qquad \forall v\in\Gamma(C_-).
\end{equation}
A tuple consisting of a spin structure, a dilaton, and a corresponding twist datum is a \emph{twisted supergravity background}.
\end{dfn}

\begin{rmk}\label{rmk:ordinary-data}
  To make a connection to the usual formulation of $\ms N=1$ supergravity (which is typically formulated in Lorentzian signature), first recall that $E$ locally has the form of Example \ref{ex:transitive-smooth} for some Lie group $G$, i.e.\
  \[E\cong_{\text{loc}}TM\oplus T^*M\oplus (M\times\mf g).\]
  The generalized metric $C_+$ can then be seen as the graph of a linear map $TM\to T^*M\oplus(M\times\mf g)$, i.e.\ it encodes a metric, 2-form, and a $G$-connection. The dilaton can be written as $\sigma=\sigma_ge^{-2\varphi}$, where $\sigma_g$ is the metric half-density and $\varphi$ is the dilaton function of $\ms N=1$ supergravity. In a similar fashion one can encode the fermionic field content, as well as the symmetries of the theory --- the latter including in particular local supersymmetry.
  
  In \cite{bvsugra,Kupka:2025ulg} the BV action for $\ms N=1$ supergravity was constructed. The Costello--Li twist \cite{CLsugra} is then obtained by expanding this action around the critical point defined by a twisted supergravity background, with all the other fields, ghosts, and antifields vanishing. The criticality conditions translate precisely to the conditions~\eqref{cond}, together with the vanishing of the generalized Ricci tensor.
\end{rmk}

Let us now analyse the consequences of these conditions. We observe that twist datum defines a subbundle $L$ of~$C_+$,  the annihilator of $e$:
\[L\coloneqq\{u\in C_+\otimes \mb C\mid \slashed ue=0\}\subset C_+\otimes \mb C,\qquad \mc L\coloneqq\Gamma(L).\]
\begin{thm}
 The subbundle $L$ is Lagrangian and involutive.
\end{thm}
\begin{proof}
Isotropy follows easily from the following consideration: if $u,v\in L$ then $0=\{\slashed u,\slashed v\}e=2\la u,v\ra e$ and thus $\la u,v\ra$ vanishes. The maximality follows from the classification of pure spinors in dimension ten~\cite{Igusa1970classification}.

Since $D$ is Levi-Civita, it preserves $C_\pm$, and so in particular we have $D_{w}\Gamma(L)\subset \Gamma(C_+\otimes\mb C)$ for any $w\in\Gamma(C_-)$. Since
\[\forall w\in\Gamma(C_-),\; u\in \Gamma(L)\qquad 0=D_w(\slashed ue)=(\cancel{D_wu})e,\]
we have that in fact $D_{C_-}$ preserves $L$. Next, using the pairing compatibility for the Courant algebroid,
\[\forall w\in\Gamma(C_-),\; u,v\in \Gamma(L)\qquad\la [u,v],w\ra=\la v,[w,u]\ra=\la v,D_wu\ra=0,\]
showing that $[\Gamma(L),\Gamma(L)]\subset\Gamma(C_+\otimes \mb C)$.

Finally, we make use of the formula~\cite[(F.11)]{Kupka:2025ulg}:
\[\ms L_ue=\tfrac12\{\di,\slashed u\}e,\qquad \forall u\in\Gamma(C_+)\]
to deduce that for all $u\in\Gamma(L)$ we have $\ms L_ue=0$. Thus, for $u,v\in\Gamma(L)$ we have
\[(\cancel{\ms L_uv})e=\ms L_u(\slashed ve)-\slashed v\ms L_ue=0,\]
hence $[u,v]\in\Gamma(L)$. This proves the involutivity of $L$.
\end{proof}
\begin{cor}
  The subbundle $L$ defines a complex structure on $M$, with $\rho(L)=T^{0,1}M\subset TM\otimes\mb C$.
\end{cor}
\begin{proof}
  For any $u\in L\cap \bar L$ we have $\bar u\in L\cap \bar L$, and so both $u+\bar u,i(u-\bar u)\in L\cap (C_+\otimes \mb R)\subset C_+\otimes \mb C$. But the latter intersection is empty, since $L$ consists of isotropic vectors and $C_+$ has a positive-definite inner product. Thus $u=0$ and $L\cap \bar L=0$, implying $\rho(L)\cap\overline{\rho(L)}=0$ due to the assumption that $\rho|_{C_+}$ is an isomorphism. The involutivity of $\rho(L)$ follows from the involutivity of $L$.
\end{proof}

\begin{rmk}
  A twisted supergravity background is equivalent to a torsion-free $SU(5)\times O(p,10+q)$-structure in the sense of $O(10+p,10+q)\times\mb R^+$ generalized geometry, analogous to the 6D statement in \cite{Ashmore:2019rkx}. Further, the subbundle $L$ corresponds to the involutive subbundle $L_{-1}$ in that reference.
\end{rmk}

\subsection{Reduction and extension}\label{subsec:red-ext}
In the preceding section we started with a smooth transitive Courant algebroid $\mc E$ over a 10-dimensional manifold $M$ equipped with a twisted supergravity background, and we showed that this induces an involutive isotropic subsheaf $\mc L\coloneqq \Gamma(L)$ inside the smooth complex Courant algebroid $\mc E\otimes\mb C$.

Let us now use $\mc L\subset\mc E\otimes\mb C$ as input in the reduction and extension procedure of \S\ref{sec:two-constructions}. Note that $\mc O$ and $\mc O^\mc L$ are the sheaves of smooth and holomorphic functions, respectively. Furthermore,
\[\mc F\cong \Gamma(L^\perp/L),\qquad \mc F^\mc L=\{\text{$\mc L$-flat sections of $\mc F$}\}.\]
\begin{lem}
  The condition $\mc F^\mc L\otimes_{\mc O^\mc L}\mc O=\mc F$ is locally satisfied.
\end{lem}
\begin{proof}
  Since $L\cong T^{0,1}M$, the flat $\mc L$-connection on $L^\perp/L$ is equivalent to a Dolbeault operator on $L^\perp/L$, i.e.\ $L^\perp/L$ is a holomorphic vector bundle over $M$. Thus its space of smooth sections is locally generated by the holomorphic ones, proving the claim.
\end{proof}

Next, we set
\[\mc O'\coloneqq(\Omega^{0,\bullet},\bar\partial),\]
and take $i\colon \mc O^\mc L\to \mc O'$ to be the inclusion.
Finally, since the Lie derivative w.r.t.\ any holomorphic vector field $x$ is a derivation of the Dolbeault complex, we can set
\[\lambda\colon \Der\mc O^\mc L\to \Der\mc O',\qquad \lambda(x)\coloneqq L_x.\]
Since this extends the action of $x$ on holomorphic functions, the property \eqref{new-commutative} is satisfied.

Starting with the smooth complex Courant algebroid $\mc E\otimes\mb C$, we can thus perform the reduction along $\mc L$ (Theorem \ref{thm:reduction}) to obtain the reduced Courant algebroid $\mc F^\mc L$ over $\mc O^\mc L$, i.e.\ a holomorphic Courant algebroid. Since $L^\perp/L\cong C_-\otimes\mb C$ and $\rho|_{C_-}$ is surjective, it follows that $\mc F^\mc L$ is transitive.

Subsequently, we apply the extension procedure along $\lambda$ (Theorem \ref{thm:courant-extension} and Proposition \ref{prop:postcompose}), resulting in the Courant algebroid $\mc E_{\text{twist}}\coloneqq\mc E'=\mc F^\mc L\otimes_{\mc O^\mc L}\mc O'$ over $\mc O'$, i.e.\ a Dolbeault resolution of the holomorphic algebroid $\mc F^\mc L$. We have thus shown the main result of this section:
\begin{thm}\label{thm:big}
  To any twist datum on a smooth transitive Courant algebroid $\mc E$ over a 10-dimensional manifold $M$ we can canonically assign a transitive Courant algebroid $\mc E_{\text{\rm twist}}$ on the Dolbeault space $M_\dbar$ (with the complex structure on $M$ defined by the twist datum).
\end{thm}

\begin{rmk}
  Recall from \S\ref{subsec:reduction} that the holomorphic structure on $\mc F\cong \Gamma(C_-\otimes\mb C)$ is defined via the flat $\mc L\cong \Gamma(T^{0,1}M)$-connection given by $D_uw$ for $u\in \Gamma(L)$ and $w\in \Gamma(C_-\otimes\mb C)$, with $D$ the generalized Levi-Civita operator defined above. Upon explicit evaluation, the resulting Dolbeault-resolved algebroid in particular recovers the complex $(\Omega^{0,\bullet}(Q),\bar D)$ of \cite{delaOssa:2014cia}, where $Q\cong C_-\otimes \mb C$.  
\end{rmk}

\subsection{Calabi--Yau example}\label{subsec:cy}
  Let $M$ be a Calabi--Yau fivefold. Suppose further that $M$ is equipped with a Ricci flat K\"ahler metric $g$ (corresponding to a  K\"ahler form $\omega$) and a chosen nonzero covariantly constant spinor $\psi$.
  
  We set $\mc E=\Gamma(TM\oplus T^*M)$ to be the standard smooth Courant algebroid, and we equip it with the spin structure such that
  \[C_\pm=\on{graph}(\pm g)=\{x\pm g(x,\slot)\mid x\in TM\},\]
  with the $\on{Spin}(5)$-structure on $C_+$ given by the one on the Calabi--Yau fivefold via the identification $C_+\cong TM$ through the anchor map. To obtain a twist datum, we take the dilaton $\sigma$ to be the metric half-density and
  \[e\coloneqq\sigma\psi.\]
  
  The associated subbundle $L$ is the preimage of $T^{0,1}M$ under $\rho|_{C_+}$, i.e.\
  \[L=\{x+g(x,\slot)\mid x\in T^{0,1}M\}=\{x+i\omega(x,\slot)\mid x\in T^{0,1}M\}=e^{i\omega}(T^{0,1}M),\]
  where we defined the action of $\omega$ on $(TM\oplus T^*M)\otimes\mb C$ by $y+\alpha\mapsto \omega(y,\cdot)$. Since $\omega$ is closed, we note that this action locally coincides with $\ms L_\beta$ for some 1-form $\beta$, and hence the action of $\omega$ preserves the bracket on $\mc E\otimes\mb C$. This in particular explains the involutivity of $L$.
  
  Next, using the anchor and $\omega$ we have the following isomorphism of complex smooth vector bundles:
  \[\tau\colon T^{1,0}M\oplus T^{*1,0}M \isom T^{1,0}M\oplus T^{0,1}M=TM\otimes\mb C \isom C_-\otimes \mb C\isom L^\perp/L,\]
  explicitly given by
  \[\tau(y+\alpha)=[(y+i\omega(y,\slot))+\tfrac i2(\omega^{-1}(\alpha)-i\alpha)]=[e^{i\omega}y+\tfrac i2e^{-i\omega}\omega^{-1}\alpha],\]
  where $\omega^{-1}\colon T^*M\otimes \mb C\to TM\otimes \mb C$ is the inverse of the map $x\mapsto \omega(x,\slot)$ and we inserted the factor $i/2$ for later convenience.
  \begin{prop}
    Under the identification $\tau$ the reduced Courant algebroid $\mc F^\mc L$ corresponds to the standard holomorphic Courant algebroid on $M$.
  \end{prop}
  \begin{proof}
    First, note that elements of $\mc F^\mc L$ correspond to $y+\alpha\in\Gamma(T^{1,0}M\oplus T^{*1,0}M)$ such that
    \[[\mc L,e^{i\omega}y+\tfrac i2e^{-i\omega}\omega^{-1}\alpha]\subset\mc L.\]
    Since $e^{\pm i\omega}$ preserves the bracket, we have that for any $\bar x\in T^{0,1}M$ the condition
    \begin{equation}\label{flatness}
      [e^{i\omega}\bar x,e^{i\omega}y+\tfrac i2e^{-i\omega}\omega^{-1}\alpha]\in \Gamma(e^{i\omega}T^{0,1}M)
    \end{equation}
    is equivalent to
    \[\Gamma(T^{0,1}M)\ni [\bar x,y+\tfrac i2e^{-2i\omega}\omega^{-1}\alpha]=[\bar x,y+\tfrac i2\omega^{-1}\alpha+\alpha].\]
    Due to the involutivity of $T^{0,1}M$ we automatically have $[\bar x,\omega^{-1}\alpha]\in\Gamma(T^{0,1}M)$ and thus the condition \eqref{flatness} is equivalent to
    \[L_{\bar x}y\in \Gamma(T^{0,1}M),\qquad L_{\bar x}\alpha=0.\]
    Asking this to hold for all $\bar x$ is then the condition of $y$ and $\alpha$ being holomorphic.
    
    It is easy to see that (with the factor $i/2$ included in the definition of $\tau$), the anchor map and pairing on (the holomorphic sections of) $T^{1,0}M\oplus T^{*1,0}M$ take precisely the required form. For the bracket we first take two holomorphic vector fields $y,z\in\Gamma(T^{1,0}M)$ and calculate
    \[[e^{i\omega}y,e^{i\omega}z]=e^{i\omega}[y,z]=\tau(L_yz),\]
    recovering the ``vector-vector'' part of the bracket on the standard holomorphic Courant algebroid. The remaining parts of the bracket are then uniquely fixed by the Courant algebroid axioms.
  \end{proof}
  \begin{cor}\label{cor:cy}
    The resulting Courant algebroid $\mc E_{\text{\rm twist}}$ is the standard Courant algebroid over $M_{\bar\partial}$. 
  \end{cor}

\section{Type I minimal BCOV theory and generalized \texorpdfstring{$G$}{G}-structures}\label{sec:last}

\subsection{Conjectures}\label{subsec:conj}
We start by observing that there are two $\mb Z/2\mb Z$-graded Batalin--Vilkovisky theories that one can assign to a twist datum (see Definition \ref{dfn:twist-datum}):
\begin{enumerate}
  \item The Costello--Li twist \cite{bvsugra,Kupka:2025ulg} of supergravity in this background, i.e.\ the $L_\infty$ algebra obtained by expanding the BV action of type I supergravity around this critical point.
  \item The Courant contact model associated to the Courant algebroid $\mc E_{\text{twist}}$ over the Dolbeault complex, following Theorem \ref{thm:big}.
\end{enumerate}

\begin{mainconj} 
The Costello--Li twist of type I supergravity is equivalent as a perturbative BV theory to the Courant contact model for $\mc E_{\text{twist}}$.
\end{mainconj}
\noindent Note that the equivalence here is understood in the strongest possible sense, i.e.\ as a quasi-isomorphism of local Lie algebras preserving the symplectic structures.
The main basis of this conjecture is the fact, proven below in Theorem \ref{thm:equivalence}, that in the case of Calabi--Yau manifolds it reduces to the original conjecture of Costello and Li \cite{CLtypeI} on the relation of minimal type I BCOV theory and the twist of type I supergravity. The Courant contact model thus provides a natural extension of the minimal type I BCOV theory to the most general twist datum. In terms of ordinary geometric data (see Remark \ref{rmk:ordinary-data}) this corresponds to twists on backgrounds with possibly nontrivial fluxes --- more precisely possibly nontrivial $2$-form, dilaton, and $G$-connection.

In addition, the present formalism allows one to regard the Main Conjecture as a comparison between two natural constructions, both set within the realm of Courant algebroids. Since the latter framework in particular provides very efficient calculational tools, it is natural to expect that it also offers the easiest and most direct way to tackle this conjecture (and thus in particular the Costello--Li conjecture).

Finally, the comparison Theorem \ref{thm:equivalence} provides an interesting insight into the origin of the minimal type~I BCOV theory, in terms of a symplectification of the contact structure on the Roytenberg--Weinstein $L_\infty$ algebra, or in terms of the moduli space of volume forms modulo generalized diffeomorphisms.

It is also natural to expect a relationship between the Courant contact model and several works which appeared recently in the literature.
Notably, in \cite{Ashmore:2025fxr} a cubic model with constrained fields was introduced based on the previous works \cite{delaOssa:2014cia,Ashmore:2018ybe} and it was shown to be equivalent (as a sheaf of $L_\infty$ algebras) to the minimal type I BCOV theory in the Calabi--Yau case.\footnote{Due to the constrained field space, this model does not define a BV theory; see Remark~\ref{rmk:resolved}. Correspondingly, the notion of equivalence is different.} The Courant contact model can thus be regarded as a version of the same story with unconstrained fields, thus satisfying the requirements of Definition \ref{dfn:local-lie-algebra}. This leads us to the following conjecture.
\begin{conj}
  The Courant contact model is quasi-isomorphic (as sheaves of $L_\infty$ algebras) to the cubic model of \cite{Ashmore:2025fxr}.
\end{conj}

Another related construction, inspired by generalized $G$-structures, is the \emph{extended BPS complex} of \cite{Kupka:2024rvl} for $G=U(5)\times SO(10+k)$, corresponding to the involutive subbundle $L$ \cite{Ashmore:2019rkx}. One immediately observes that the vector space underlying this BPS complex coincides with the space of fields \eqref{eq:M}, naturally leading to:
\begin{conj}\label{conj:bps}
  The complex underlying the Courant contact model is quasi-isomorphic to the extended BPS complex for $G=U(5)\times SO(10+k)$.
\end{conj}
\noindent The extended BPS complex was constructed in an ad-hoc fashion, by appending $\mc O^!$ as an extra row in the discussion of the superpotentials in \cite{Kupka:2024rvl}. The present construction would thus provide a natural geometric rationale for this procedure, as the symplectification of a contact structure.

\begin{rmk}
One could imagine the supersymmetry conditions for the background as being split into two sets, analogous to the F-term and D-term conditions of the six-dimensional construction in~\cite{Ashmore:2019rkx} (which are named in this way because they are the origin of the corresponding conditions in the four-dimensional $\ms N = 1$ supergravity theories obtained by compactification on them). 
In the ten-dimensional case, the F-terms would naturally be identified as those whose variations appear in the part of \eqref{eq:M} in the relevant degree.  
The Courant contact model should thus be related to the formal moduli problem for backgrounds satisfying these F-term supersymmetry conditions in the same way that Baulieu's theory in~\cite{BaulieuCS} (the holomorphic twist of $\N=(1,0)$ super Yang--Mills theory in ten dimensions) is related to the corresponding moduli problem for holomorphic bundles.
\end{rmk}

\subsection{Review of minimal type I BCOV theory}
We recall the definition of minimal type I BCOV theory on a Calabi--Yau fivefold $(M,\Omega)$, i.e.\ a complex 5-dimensional manifold $M$ with a holomorphic volume form $\Omega$. Our theory is a one-dimensional central extension of the original theory studied in~\cite{CLtypeI}, and admits a nondegenerate odd symplectic structure. For details on  this particular ``potential theory,'' see~\cite[\S5.4.1]{RSW}, as well as~\cite{CLtypeI} for the original construction and detailed discussion.

First, we introduce the holomorphic polyvector fields
\[\PV^{i,\bullet}\coloneqq\Omega^{0,\bullet}(\Lambda^iT^{1,0}M).\]
Contraction with the holomorphic volume form $\Omega$ induces an isomorphism 
    \begin{equation}\label{pv-vs-omega}
        \begin{tikzcd}
            \PV^{i,\bu} \arrow[r,"\cong"] & \Omega^{n-i,\bu}
        \end{tikzcd}
    \end{equation}
The space of fields of the BCOV theory is then described by the following sheaf of $\mb Z/2\mb Z$-graded locally free modules over the Dolbeault complex $\mc O=(\Omega^{0,\bullet},\bar\partial)$:
\begin{equation}\label{bcov-chain}
  \Pi\ms M_{\text{BCOV}}\coloneqq\Omega^{0,\bu} \oplus \Pi(\Omega^{1,\bu} \oplus \PV^{1,\bu}) \oplus \PV^{0,\bu}.
\end{equation}
We will call the individual fields $(\beta, \gamma, \mu, \nu)$.\footnote{Note that the antiholomorphic form part further shifts the parity as appropriate, so that e.g.\ the $\Omega^{1,3}$-part of $\gamma$ is even.}

 This theory is odd-shifted symplectic, with $\beta$ being the antifield of~$\nu$ and $\gamma$ the antifield of~$\mu$. To see that this pairing is odd, it is important to recall that we are on a fivefold, so that the wedge-and-integrate pairing on Dolbeault forms has odd degree. The pairing and action are defined by integrating against the fixed background holomorphic volume form $\Omega$:
    \begin{equation}\label{bcov}
    \omega = \int \Omega \left( \delta \beta\, \delta \nu + \delta \mu \vee \delta \gamma \right),\qquad S = S_{\bar\partial}+\int \Omega \left( \lo_\mu \beta + \frac12 \frac{1}{1+\nu} i_\mu i_\mu \del \gamma \right),
    \end{equation}
    where $S_{\bar\partial}$ is the functional whose Hamiltonian vector field is the operator $\bar\partial$ on the field space and, following the literature, we use $\vee$ to denote the pairing of vectors with 1-forms. Finally, $\lo$ corresponds to the ``holomorphic Lie derivative'' (e.g.\ if $x\in \Gamma(T^{1,0}M)$ then on forms $\lo_x= i_x\partial+\partial i_x$), which has been extended naturally to the Dolbeault resolutions of the corresponding objects (cf.\ Example \ref{ex:long}).
    \begin{rmk}
      In particular it follows that the cochain complex structure on $\Pi\ms M_{\text{BCOV}}$ is given by
    \begin{equation*}
        \begin{tikzcd}
            \Omega^{0,\bu} \arrow[r,"\del"] & \Omega^{1,\bu} \oplus \PV^{1,\bu} \arrow[r,"\div"] & \PV^{0,\bu},
        \end{tikzcd}
    \end{equation*}
    where the \emph{divergence operator} $\div$ is defined with reference to the isomorphism \eqref{pv-vs-omega}.
    \end{rmk}
    
\subsection{Equivalence between the Courant contact model and the BCOV theory}
Let us first recall the Courant contact model for the standard Courant algebroid on $M_{\bar\partial}$ on a Calabi--Yau fivefold $(M,\Omega)$ (cf.\ Corollary \ref{cor:cy}). The field space is
\[\Pi\ms M=\Omega^{0,\bullet}\oplus \Pi(\Omega^{1,\bullet}\oplus \PV^{1,\bullet})\oplus \Omega^{5,\bullet},\]
with the fields denoted $(f,\alpha,x,\lambda)$.
\begin{lem}
  The symplectic form and the action functional of the Courant contact model for the standard Courant algebroid on $M_{\bar\partial}$ are (here again $S_{\bar\partial}$ denotes the functional generating the $\bar\partial$ operator)
  \[\omega=\int_M\delta\lambda(\delta f+\tfrac12x\vee \delta\alpha+\tfrac12\delta x\vee\alpha)-\lambda\delta x\vee\delta\alpha,\qquad S=S_{\bar\partial}+\int_M\lambda\left(\lo_x f+\tfrac12 x\vee \lo_x\alpha\right).\]
\end{lem}
\begin{proof}
  The only non-immediate part is the last term in $S$, in which we use $i_{\lo_x x}=\lo_x i_x+i_x \lo_x$ to write
  \[\la \xi,\ms L_\xi\xi\ra=\alpha\vee \lo_xx+x\vee (\lo_x\alpha+i_x\partial\alpha)=(i_{\lo_xx}+i_x\lo_x+i_xi_x\partial)\alpha=(\lo_x i_x+3i_x\lo_x-i_x\partial i_x)\alpha=3i_x\lo_x\alpha.\qedhere\]
\end{proof}
  \begin{thm}\label{thm:equivalence}
    Let $(M,\Omega)$ be a Calabi--Yau fivefold. Then the field redefinition from $(f,\alpha,x,\lambda)$ to $(\beta,\gamma,\mu,\nu)$ given by
    \[f=\beta+\frac{\mu\vee\gamma}{2(\nu-1)},\qquad \alpha=-\gamma,\qquad x=\frac\mu{\nu-1},\qquad \lambda=\Omega(\nu-1)\]
    defines an equivalence of BV theories between the Courant contact model for the standard Courant algebroid over $M_{\bar\partial}$ and the minimal type I BCOV theory.
  \end{thm}
  
  \begin{proof}
    This is a direct calculation. In particular the holomorphy of the transformation ensures that the $\bar\partial$-generating parts $S_{\bar\partial}$ of the corresponding functionals map into each other. The field redefinitions define an $L_\infty$ quasi-isomorphism between the corresponding local formal moduli problems. While they do \emph{not} define a canonical transformation, one checks straightforwardly that they intertwine the respective odd symplectic structures.
  \end{proof}
  \begin{rmk}
    We note that the theory given by \eqref{bcov-chain} and~\eqref{bcov} can be defined in principle on any Calabi--Yau $n$-fold $M$ and gives a shifted symplectic structure with parity of the symplectic form coinciding with the parity of $n$. 
    (We emphasize, though, that this theory only agrees with type I BCOV theory when $n=5$.)
    Similarly, the field space of the Courant contact model for the standard Courant algebroid over $M_{\bar\partial}$ is given by
    \[\Pi\ms M=\Omega^{0,\bullet}\oplus \Pi(\Omega^{1,\bullet}\oplus \PV^{1,\bullet})\oplus \Omega^{n,\bullet},\]
    with the symplectic form and action
    \[\omega=\int_M\delta\lambda(\delta f+\tfrac12x\vee \delta\alpha+\tfrac12\delta x\vee\alpha)+(-1)^n\lambda\delta x\vee\delta\alpha,\qquad S=S_{\bar\partial}+\int_M\lambda\left(\lo_x f+\tfrac12 x\vee \lo_x\alpha\right).\]
    The equivalence of the two models is then established via
    \[f=\beta-\frac{\mu\vee\gamma}{2(1-\nu)},\qquad \alpha=(-1)^n\gamma,\qquad x=(-1)^n\frac\mu{1-\nu},\qquad \lambda=(-1)^n\Omega(1-\nu).\]
     \end{rmk}

\printbibliography

\appendix

\end{document}